\documentclass[10pt,a4paper]{article}
\usepackage[utf8]{inputenc}
    \usepackage[round,sort,comma,numbers,authoryear]{natbib}
    \setcitestyle{aysep={ },yysep={,}}
\usepackage[left=2.5cm, right=2.5cm, bottom=3cm, top=3cm]{geometry}	
\usepackage{lmodern}
\usepackage{amsmath}
\usepackage{amsthm}
\usepackage{amsfonts}
\usepackage{amssymb}
\usepackage{array}
\usepackage{hyperref}
\usepackage{makeidx}
\usepackage{graphicx}
\graphicspath{{}{figures/}{../figures/}}
\usepackage{float}
\floatstyle{boxed} 
\restylefloat{figure}
\usepackage{esint}
\usepackage[table]{xcolor}
\usepackage{pdfpages}
\usepackage[title]{appendix}
\usepackage{enumerate}
\usepackage[title]{appendix}
\usepackage{makeidx}

\author{Claudio Bellani\footnote{Dept. of Mathematics, Imperial College London.}, Damiano Brigo\footnotemark[\value{footnote}]}
\title{Mechanics of good trade execution in the framework of linear temporary market impact}
\date{Wednesday 8 July 2020}

{\newtheorem{thm}{Theorem}[section]
\newtheorem{prop}[thm]{Proposition}
\newtheorem{lemma}[thm]{Lemma}
\newtheorem{corol}[thm]{Corollary}
{\theoremstyle{definition}{
\newtheorem{defi}[thm]{Definition}
\newtheorem{assumption}[thm]{Assumption}
\newtheorem{remark}[thm]{Remark}

}}}

\usepackage{mathrsfs}
\usepackage{amsfonts}
\usepackage{xfrac}

\usepackage{amsthm}
\usepackage{amsmath}
\usepackage{amssymb}

\newsavebox{\fminipagebox}
\NewDocumentEnvironment{fminipage}{m O{\fboxsep}}
{\par\kern#2\noindent\begin{lrbox}{\fminipagebox}
		\begin{minipage}{#1}\ignorespaces}
		{\end{minipage}\end{lrbox}%
	\makebox[#1]{%
		\kern\dimexpr-\fboxsep-\fboxrule\relax
		\fbox{\usebox{\fminipagebox}}%
		\kern\dimexpr-\fboxsep-\fboxrule\relax
	}\par\kern#2
}

\newcommand{\R}{\mathbb{R}}
\newcommand{\Rd}{\R^{d}}
\newcommand{\Rn}{\R^{n}}

\newcommand{\half}{\frac{1}{2}}
\newcommand{\inverse}{^{-1}}

\newcommand{\abs}[1]{\left\lvert {#1} \right\rvert}

\newcommand{\transpose}{^{\mathsf{T}}}
\newcommand{\squared}{^{2}}

\newcommand{\argmin}{\mathrm{argmin}}

\newcommand{\intzerot}{\int_{0}^{t}}
\newcommand{\intZeroTimeHorizon}{\int_{0}^{\timeHorizon}}

\newcommand{\dotEta}{\dot{\eta}}

\newcommand{\simplex}{\lbrace (s,t) \in \R\squared:\,  0 \leq s \leq t \leq \timeHorizon \rbrace}

\newcommand{\airyFirstFunction}{\mathtt{Ai}}
\newcommand{\airySecondFunction}{\mathtt{Bi}}

\newcommand{\derivative}{^{\prime}}

\newcommand{\partition}{\pi}
\newcommand{\timeHorizon}{T}
\newcommand{\timeWindow}{{[0,\timeHorizon]}}

\newcommand{\norm}[1][\cdot]{\left\lVert {#1}\right\rVert}

\newcommand{\pvarNormInterval}[3][\cdot]{\lVert {#1} \rVert_{{#2}\text{-var}, {#3}}}

\newcommand{\smoothCompactlySupportedFunctions}{C^{\infty}_{0}}

\newcommand{\sobolevSpace}[1][1,2]{W^{#1}}
\newcommand{\sobolevSpaceOneTwo}{\sobolevSpace}
\newcommand{\sobolevSpaceCompactSupport}[1][1,2]{\sobolevSpace[#1]_{0}}
\newcommand{\sobolevSpaceOneTwoCompactSupport}{\sobolevSpaceOneTwo_{0}}

\newcommand{\Ltwo}{L\squared}

\newcommand{\Prob}{{\mathbb{P}}}

\newcommand{\Expectation}{\mathbb{{E}}}
\newcommand{\Variance}{\mathrm{Var}}

\newcommand{\sigmaAlgebra}{\mathfrak{F}}

\newcommand{\probabilitySpace}{\big(\Omega,\sigmaAlgebra, \Prob \big)}

\newcommand{\brownianMotion}{W}
\newcommand{\boundedVariationPart}{A}
\newcommand{\martingale}{M}
\newcommand{\semimartingale}{S}

\newcommand{\spaceInventoryTrajectories}{\mathcal{Q}}
\newcommand{\fundamentalPrice}{S}
\newcommand{\priceProcess}{S}

\newcommand{\inventory}{q}

\newcommand{\initialInventory}{\mathtt{x_0}}
\newcommand{\initialInventoryAtTimeT}{\mathtt{x}_t}
\newcommand{\liquidationTarget}{\mathtt{x}_{\mathtt{T} } }
\newcommand{\inventoryRate}{\dot{\inventory}}

\newcommand{\spaceExecutionRates}{\mathcal{R}}
\newcommand{\coeffMarketImpact}{c_{1}}
\newcommand{\coeffRiskAversion}{c_{2}}

\newcommand{\lagrangian}{L}
\newcommand{\Lagrangian}{\lagrangian}
\newcommand{\Fnorm}[1]{\lVert {#1} \rVert_{F}}
\newcommand{\normInventoryTrjectories}[1]{\lVert{#1}\rVert_{\coeffMarketImpact, \coeffRiskAversion}}
\newcommand{\ratioAversionOverImpact}{c_{3}}

\newcommand{\costFunctional}{J}
\newcommand{\spaceInventoryTrajectoriesPathwise}{\spaceInventoryTrajectories_{\text{pw}}}
\newcommand{\spaceInventoryTrajectoriesFuel}{\spaceInventoryTrajectories_{\text{fuel}}}
\newcommand{\spaceInventoryTrajectoriesStatic}{\spaceInventoryTrajectories_{\text{static}}}

\newcommand{\spaceUnbiasedInventoryTrajectories}{\mathcal{U}}
\newcommand{\spaceUnbiasedInventoryTrajectoriesInitialConstraint}{\spaceUnbiasedInventoryTrajectories^{0,\initialInventory} }
\newcommand{\stateVariable}{X}
\newcommand{\actualPricePath}{x }

\newcommand{\coeffPenalisationOutstandingInventory}{c_{5}}
\newcommand{\ratioTerminalPenalisationOverImpact}{c_{6}}

\makeindex
\begin{document}
\numberwithin{equation}{section}

\maketitle

\vspace{0.5cm}

\begin{quotation}\begin{small}
		\noindent	\textbf{Abstract.} We define the concept of good trade execution and we construct explicit adapted good trade execution strategies in the framework of linear  temporary market impact. Good trade execution  strategies   are dynamic, in the sense that they react to the actual realisation of the traded asset price path over the trading period; this is paramount in volatile regimes, where price trajectories can considerably deviate from their expected value. Remarkably however, the implementation of our strategies does not require the full specification of an SDE evolution for the traded asset price, making them robust across different models. Moreover,  rather than minimising the \emph{expected} trading cost, good trade execution strategies minimise  trading costs in a \emph{pathwise} sense, a point of view not yet considered in the literature. The mathematical apparatus for such a pathwise minimisation hinges on 
		certain random Young differential equations that correspond to the Euler-Lagrange equations of the classical Calculus of Variations. These Young differential equations characterise our good trade execution strategies in terms of an initial value problem that allows for easy implementations. 
 \end{small}
\end{quotation}

\vspace{0.5cm}

%
%
%



\section{Introduction}

Executions of large trades can affect the price of the traded asset, a phenomenon known as \emph{market impact}. The price is affected in the direction unfavourable to the trade: while selling, the market impact decreases the price; while buying, the market impact increases the price.    Therefore, a trader who wishes to minimise her trading costs has to split her order into a sequence of smaller sub-orders which are executed over a finite time horizon. How to optimally split a large order is a question that naturally arises.

Academically, the literature discussing such an optimal split  was initiated by the seminal papers by \cite{AC01opt} and by \cite{BL98opt}. Both papers deal with 
the trading process of one large market participant  who would like to buy or sell a large amount of shares or contracts during a specified duration. 
The optimisation problem is formulated as a trade-off between two pressures. On the one hand, market impact  demands to trade slowly in order to minimise the unfavourable impact that the execution itself has on the price. 
On the other hand, traders have an incentive to trade rapidly, because they do not want to carry the risk of adverse price movements away from their decision price. Such a  trade-off between market impact and market risk is usually translated into a stochastic control problem where the trader's strategy (i.e. the control) is the trading speed. The class of admissible strategies defines  the set over which the risk-cost functional is optimised. 

In the design of mathematical models for optimal trade execution we identify two phases. The first phase is the description of trading costs. This refers to the choice of a function $F$ that depends on time, asset price, quantity to execute and trading speed, and models the instantaneous cost of trading. The overall cost during the time window $\timeWindow$ is then expressed as the time integral
\[
\costFunctional(\inventory) = \int_{0}^{\timeHorizon}
F(t,\actualPricePath_t,\inventory_t,\dot{\inventory}_t)dt,
\]
where the path $t \mapsto \actualPricePath_t$ is the evolution of the asset price during the trading period. The letter $\inventory$ stands for quantity of the asset and the trajectory $t\mapsto \inventory(t)$, $\timeWindow \rightarrow \R$, is referred to as inventory trajectory. Its time derivative $\dot{\inventory}$ is the rate of execution and it represents the control variable that a trader modulates while executing the trade. 

The minimisation of the trading cost $\costFunctional$ faces the challenge that the price path $t \mapsto \actualPricePath_t$ is not known at the beginning of the trading period. Hence, in order to gain some predictive power, a stochastic model for the evolution of the asset price is introduced. This is the second phase in the design of  mathematical models for trade execution. Concretely, it means that a stochastic process $\lbrace \fundamentalPrice_t: \, 0\leq t \leq \timeHorizon \rbrace$ is introduced and the actual price trajectory $(x_t)$ is thought of as a realisation of this stochastic process. Then, the mathematical optimisation focuses on the expected trading cost
\begin{equation}\label{eq.expectedTradingCost}
\Expectation\left[
\intZeroTimeHorizon F(t,\fundamentalPrice_t, \inventory_t,\inventoryRate_{t}) dt
\right]. 
\end{equation}
Notice that this entails a considerable degree of model dependency, in that the optimisation is based on the distributional assumptions on the price process. 

Two alternatives exist for the minimisations of the expected trading cost in equation \eqref{eq.expectedTradingCost}. These alternatives are static minimisation (giving rise to static trading strategies) and dynamic minimisation (giving rise to dynamic trading strategies).

Static strategies are  completely decided at the beginning of the trading period; they are based only on the information available  at the initial time of the trade. Mathematically, this is formulated by considering $\inventory$ as a deterministic path. In this case it is often observed that,  by interchanging expectation and time integral in equation \eqref{eq.expectedTradingCost}, the actual realisation of the price process disappears from the formulae, replaced by its expected trajectory.  When the expected price path is the only feature of the price process that enters the formulae (as in \cite{AC01opt}), the static strategy does not take into account the volatility of asset prices, whose role however is paramount in financial markets.  A visual representation of the relevance of volatility in the context of trade execution is provided by Figure \ref{fig.aPosterioriTwoVolatilities}.

In Figure \ref{fig.aPosterioriTwoVolatilities} Almgren and Chriss's framework is adopted. The price process is  a standard one-dimensional Brownian motion and two price paths are considered, one with low volatility and the other with high volatility. Notwithstanding the remarkable difference between the two, they have the same expected path (dashed blue line in the first quadrant) and, as a consequence, the static liquidation strategy is the same for both price paths (dashed blue line in the second quadrant).  The simplicity of the model is such that it compromises on the possibility to distinguish rather different market regimes. This is made clear by comparing the static optimal solution with the a-posteriori one.

\begin{figure}
\centering
\includegraphics[width=0.50\textwidth]{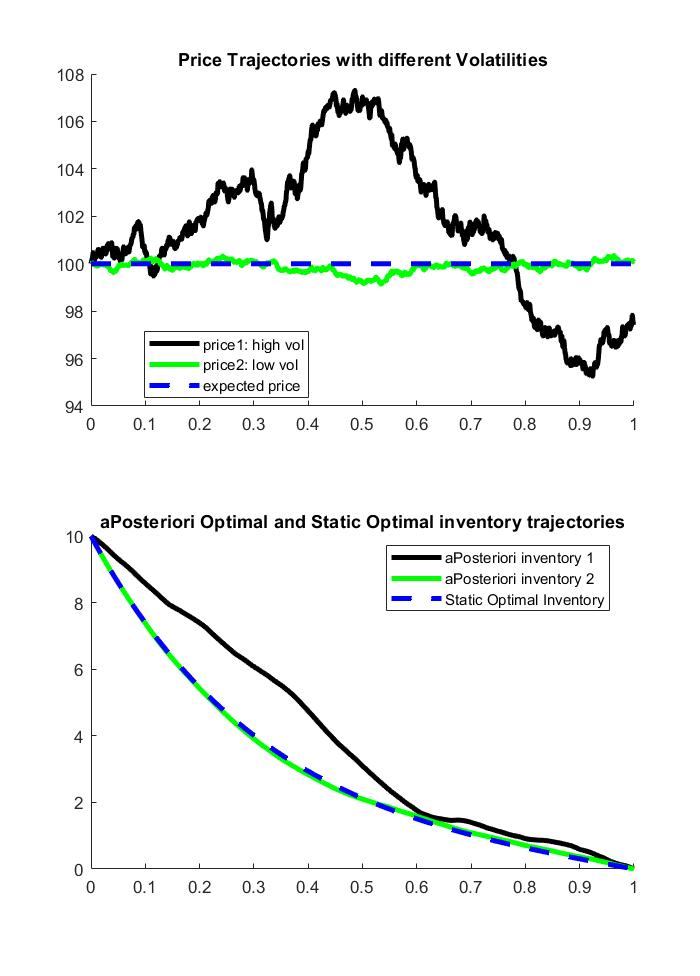}
\caption{{A-posteriori optimal and static optimal inventories in two different volatility regimes}}
\label{fig.aPosterioriTwoVolatilities}
\end{figure}

The a-posteriori solution  is the minimiser $\inventory$ of the cost functional $\costFunctional$ given the actual price trajectory $\actualPricePath$. This is not implementable in real trading because it is anticipative, in that it assumes that the entire price trajectory is known at the beginning of the trading period. However, since it is independent of the choice of the price process,  the a-posteriori solution constitutes a useful term of comparison for the stochastic model. In the example of Figure \ref{fig.aPosterioriTwoVolatilities}, we observe how different the two  a-posteriori solutions corresponding to the two market regimes are. In the case of low volatility, the a-posteriori solution is close to the static one, because the price path does not depart significantly from its expected trajectory. Instead, in the case of high volatility,  the a-posteriori solution deviates from the static one: the inventory trajectory is considerably steeper where the price is above its expected value, and it is almost flat when the price is below its expected value.

In order to take into account more features of the price process (such as its volatility), the literature on optimal trade execution has utilised the mathematical techniques of stochastic optimal control. This has produced the second alternative the minimisation of the expected cost in equation \eqref{eq.expectedTradingCost}, and dynamic trading strategies proliferated since  \cite{BL98opt} (in discrete time) and  \cite{GS11opt} (in continuous time). An excellent presentation of the techniques of stochastic optimal control applied to trade execution is contained in the textbook by \cite{CJP15alg}. 

Dynamic trading strategies take fully into account the distributional features of the price process because they are obtained via the Hamilton-Jacobi-Bellman equation, in which the generator of the diffusion that models the price enters.\footnote{In the case of linear temporary market impact and quadratic inventory cost, a recent work by \cite{BMO18opt} actually discusses techniques that can be more generally applied to the case of general semimartingales. In this case there is no HJB equation; instead the authors rely on forward-backward stochastic differential equations. In Section \ref{sec.framework}, we will review this general solution.}
Furthermore, dynamic strategies are  random when seen from the initial time, in that they depend   on the information that is revealed to the trader during the trading period. Mathematically, this means that dynamic strategies are stochastic processes adapted to the relevant market information filtration.  Since deterministic strategies are in particular adapted stochastic processes, the class of static strategies is a subset of the class of dynamic strategies. Therefore, the minimisation  over the class of dynamic strategies is expected to improve the result obtained when minimising over the smaller class of  static strategies.  

This however is not always confirmed in the models. Indeed, despite the mathematical sophistication, cases exist in which  optimal trading strategies, although sought among dynamic ones, are in fact static. One of such cases is for example the ``Liquidation without penalties only temporary impact'' in  \cite[Section 6.3]{CJP15alg}, an other is the ``Optimal acquisition with terminal penalty and temporary impact'' in  \cite[Section 6.4]{CJP15alg}. This reduction to static optimal solutions clashes with the intuition for which trading strategies should take into account actual realisations of price paths, as the a-posteriori solutions in Figure \ref{fig.aPosterioriTwoVolatilities} suggest. 

A second drawback of applying  the technique of HJB equation to the problem of optimal trade execution is the heavy model dependence. Optimality of the trading strategies holds under the assumption that the price follows some specified dynamics, and this invests  of considerable importance the second phase in the design of mathematical models.  

In this paper, we propose a new alternative for the minimisation of trading costs. This new alternative considers the pathwise optimisation of the cost functional $\costFunctional$ without taking expectation. We observe that the reason for the anticipativeness of a-posteriori solutions is the imposition of the constraint that the liquidation terminates exactly at the (arbitrarily fixed) trading horizon. Relaxing this constraint enables to produce adapted pathwise solutions that display two remarkable features. On the one hand, they avoid the degeneracy to static trajectories even in the cases where the techniques of HJB equation do not produce genuinely dynamic strategies; on the other hand, their model dependence is moderate and confined to the expected trajectory of the price path, as was the case for static strategies, rather than to the full law of the price process.

Our trading strategies give rise to inventory trajectories that are obtained in closed-form formulae. Moreover, we can characterise these trajectories as solutions to certain random Young differential equations, inspired   by the second-order Euler-Lagrange equations in the classical Calculus of Variations. Such a characterisation allows to implement the inventory trajectories via an easily-simulated initial value problem.

\nocite{CJ19alg}
\nocite{HJN19mea}
\nocite{CDJ17alg}
\nocite{CJ19tra}
\nocite{MC20mar}
\nocite{RZ18gam}
\nocite{CDP20max}
\nocite{KK18nas}
\nocite{XZ13opt}
\nocite{DS86use}

The rest of the paper is organised as follows. Section \ref{sec.framework} describes in detail the mathematical framework in which the problem of optimal trade execution is formulated. Our descriptions examines in particular two aspects of the mathematical models. The first aspect (Section \ref{sec.reductionStatic}) is the reduction to static optimal inventories that happens in the context of stochastic optimal control of the expected quantity in equation \eqref{eq.expectedTradingCost}.  Proposition \ref{prop.reductionToStaticSolution} examines such a reduction, listing its causes. This is novel in the literature and answers the questions raised in \cite{BD14opt}, \cite{BP18sta} and \cite{BBDN18sta}  about the comparison between static and dynamic solutions to the problem of optimal trade execution.  The second aspect is  the unbiasedness of liquidation errors (Section \ref{sec.errorsOfLiquidation}).  

Section \ref{sec.goodTradeExecutions} presents the concept of good trade execution. Section \ref{sec.IC} specialises good trade executions in the case of linear temporary market impact and quadratic inventory cost. In particular, Section \ref{sec.closedformIC} derives a closed-form formula for good trade executions, and Section \ref{sec.eulerLagrangeIC} characterises it in terms of a Cauchy problem with random Young differential equations. Uniqueness of the good trade execution follows from this characterisation.

Section \ref{sec.alternativeRiskCriteria} presents good trade executions with risk criteria other than the quadratic inventory cost. More precisely, Section \ref{sec.timeIC} considers a time-dependent variant of the quadratic inventory cost, whereas Section \ref{sec.varInspiredRiskCriterion} presents good trade executions when the risk criterion is inspired by the value-at-risk adopted in \cite{GS11opt}.    

Finally, two applications are given in Section \ref{sec.applications}, and Section \ref{sec.conclusions} concludes the paper. Appendix \ref{sec.eulerLagrangeInPresenceOfPricePath} presents the mathematical apparatus on which the characterisation of good trade executions is based.

\section{Framework} \label{sec.framework}
We adopt the perspective of liquidation; the case of acquisition is \emph{mutatis mutandis} the same. Let $\initialInventory$ denote initial inventory, and let $\liquidationTarget=0$ be the liquidation target. The letter $\inventory$ stands for quantity of the asset and the trajectory $t\mapsto \inventory(t)$, $\timeWindow \rightarrow \R$, shall be referred to as inventory trajectory. Its time derivative $\dot{\inventory}$ is the rate of execution and it represents the control variable that a trader modulates while executing the trade. Without yet referring to any probabilistic structure, let us introduce the space of such inventory trajectories:
\begin{equation*}
\begin{split}
\spaceInventoryTrajectories^{0,\initialInventory}_{\text{pw}} := \Big\lbrace
\inventory:\timeWindow\rightarrow & \R,  \quad  \inventory \text{ absolutely continuous},
\, \, 
q(0) = \initialInventory, \, q(\timeHorizon) = \liquidationTarget\Big\rbrace.
\end{split}
\end{equation*}
The subscript ``pw'' stands for ``pathwise'' and emphasises the  non-probabilistic perspective. 

\begin{defi}\label{def.priceprocess}
 Let $\probabilitySpace$ be a probability space, and let  $\lbrace \fundamentalPrice_t:\, 0\leq t\leq T\rbrace$ be a stochastic process defined on it. We say that $\fundamentalPrice$ is a price process if:
 1. for all $0\leq t\leq T$ the second moment of $\fundamentalPrice_t$ is finite;
 2. the maps $t\mapsto \Expectation[\fundamentalPrice_t]$ and $t\mapsto \Expectation[\fundamentalPrice\squared_t]$ are in $L^1[0,T]$;
 3. there exists some $p\geq 1$ such that all the paths of $\fundamentalPrice$ are of finite $p$-variation, i.e. for all $\omega$ in $\Omega$, 
\begin{equation*}
\pvarNormInterval[\fundamentalPrice_{\cdot} (\omega)]{p}{\timeWindow} < \infty. 
\end{equation*} 
\end{defi}


Notice that the paths of the price process are not necessarily assumed to be continuous.

Given a price process  $\lbrace \fundamentalPrice_t:\, 0\leq t\leq T\rbrace$, we let $\lbrace \sigmaAlgebra_t: \, 0\leq t\leq T\rbrace$ be the minimal $\Prob$-completed right-continuous filtration generated by $\fundamentalPrice$. It is always assumed that $\sigmaAlgebra_0$ is trivial. 

If the price process is a semimartingale, we additionally introduce the following terminology. We say that the semimartingale $\semimartingale$ is a totally square integrable special semimartingale if the following two conditions hold:
\begin{enumerate}
\item the semimartingale $\semimartingale$ is a special semimartingale, i.e. it admits a canonical decomposition 
\begin{equation*}
\semimartingale_t = \semimartingale_0 + \boundedVariationPart_t + \martingale_t, \qquad 0\leq t \leq \timeHorizon,
\end{equation*}
where $\boundedVariationPart$ is a predictable bounded variation process,  $\martingale$ is a local martingale, and $\boundedVariationPart_0 = \martingale_0 = 0$;
\item the following integrability holds:
\begin{equation*}
\Expectation \Big[ \langle \martingale \rangle_\timeHorizon  \Big] + \Expectation \left[ \norm[\boundedVariationPart]_{2,\timeWindow}\squared  \right] \, < \infty,
\end{equation*}
where $\langle \martingale \rangle$ denotes the quadratic variation of the local martingale $\martingale$, and $ \norm[\boundedVariationPart]_{2,\timeWindow}$ denotes the $2$-variation of the path $\boundedVariationPart$ on the time interval $\timeWindow$.\footnote{Recall that the 2-variation   $ \norm[x]_{2,\timeWindow}$ of a path $x: t\mapsto x_t \in \Rd$ is defined as 
\begin{equation*}
 \norm[x]_{2,\timeWindow}
 :=
 \sup
 \left\lbrace
 \left(
 \sum_{t_{i}} \abs{x_{t_{i+1}} -x_{t_{i}}}^{2}
 \right)^{1/2} : 
 \, \, 
 0=t_0< t_1<\dots<t_n=\timeHorizon
 \right\rbrace,
\end{equation*}
where the supremum is taken over all the partitions of the interval $\timeWindow$.}
\end{enumerate}

Execution rates are progressively measurable square-integrable processes; more precisely, we define the space of execution rates as
\begin{equation}\label{eq.definitionOfSpaceExecutionRates}
\begin{split}
\spaceExecutionRates:= \Big\lbrace
r \in \Ltwo\Big([0,T]\times& \Omega, \, \mathcal{B}[0,T]\otimes \sigmaAlgebra_{\timeHorizon}, \, dt\otimes \Prob\Big): \, \\
&
r \text{ is } (\sigmaAlgebra_t)_t\text{-progressively measurable}\Big\rbrace.
\end{split}
\end{equation}
Notice that the measurability depends on the filtration of the price process. 

Admissible inventory trajectories are first integrals of execution rates with initial value $\initialInventory$. More precisely, we define the space $\spaceInventoryTrajectories^{0,\initialInventory}$ of admissible inventory trajectories as 
\begin{equation}\label{eq.definitionOfSpaceInventoryTrajectories}
\spaceInventoryTrajectories^{0,\initialInventory} 
:=
\left\lbrace
(\inventory_t)_{t\in\timeWindow}: \,\, 
\exists r \in \spaceExecutionRates, \, 
\inventory_t=\initialInventory + \int_{0}^{t}r_u du \, \forall 0\leq t \leq \timeHorizon 
\right\rbrace.
\end{equation}
Among admissible inventory trajectories we distinguish those that are fuel-constrained, namely such that their terminal value is $\liquidationTarget=0$. Thus, a fuel-constrained admissible inventory trajectory is an $(\sigmaAlgebra_t)_t$-adapted process with absolutely continuous paths, with deterministic initial value $\initialInventory$,  terminal value $\liquidationTarget=0$, and such that its derivative is in $\spaceExecutionRates$.
More precisely, we define the space $\spaceInventoryTrajectoriesFuel^{0,\initialInventory}$ of fuel-constrained admissible inventory trajectories as 
\begin{equation*}
\spaceInventoryTrajectoriesFuel^{0,\initialInventory}
:=
\left\lbrace
(\inventory_t)_{t\in\timeWindow} \in \spaceInventoryTrajectories^{0,\initialInventory}: \quad 
 \inventory_\timeHorizon = \liquidationTarget \,\, \Prob\text{-a.s.}
\right\rbrace.
\end{equation*}
Notice that every realisation of a generic $\inventory$ in $\spaceInventoryTrajectoriesFuel^{0,\initialInventory}$ is a path in $\spaceInventoryTrajectories^{0,\initialInventory}_{\text{pw}}$, namely for all $\inventory$ in $\spaceInventoryTrajectoriesFuel^{0,\initialInventory}$ and all $\omega$ in $\Omega$ it holds
\begin{equation*}
\left(\inventory_t(\omega)\right)_{0\leq t \leq \timeHorizon} \, \in \, \spaceInventoryTrajectories^{0,\initialInventory}_{\text{pw}}.
\end{equation*}
In the space of fuel-constrained inventory trajectories we isolate the subspace of static trajectories, given by
\begin{equation*}
\spaceInventoryTrajectoriesStatic^{0,\initialInventory} = 
\left\lbrace \inventory \in \spaceInventoryTrajectoriesFuel^{0,\initialInventory}: \quad 
\inventory_t \text{ is } \sigmaAlgebra_0\text{-measurable} \text{ for all } t \geq 0 \right\rbrace.
\end{equation*}
These are the execution strategies whose entire trajectories are $\sigmaAlgebra_0$-measurable, namely deterministic. We say that the admissible inventory trajectories not in $	\spaceInventoryTrajectoriesStatic^{0,\initialInventory}$ are non-static (or dynamic): therefore, the admissible inventory trajectory $\inventory$ is non-static if $\inventory$ is in  $\spaceInventoryTrajectories^{0,\initialInventory}\setminus \spaceInventoryTrajectoriesStatic^{0,\initialInventory}$.

It is convenient to extend the definitions of the spaces of inventory trajectories to the case where the initial time is not zero. The symbols $\spaceInventoryTrajectoriesPathwise^{t,\initialInventoryAtTimeT}$, $\spaceInventoryTrajectories^{t,\initialInventoryAtTimeT}$, $\spaceInventoryTrajectoriesFuel^{t,\initialInventoryAtTimeT}$ and $\spaceInventoryTrajectoriesStatic^{t,\initialInventoryAtTimeT}$ will denote the straightforward generalisations of the definitions above to the case where the initial time is $t$ in $[0,\timeHorizon)$ and the trajectories are pinned to the value $\initialInventoryAtTimeT$ at time $t$. 

With the notation introduced so far, we now formulate the classical stochastic optimisation problem associated with optimal trade execution. 

Let $\stateVariable = (\fundamentalPrice,\inventory)$ denote the state variable, which keeps track of the fundamental price $\fundamentalPrice$ and of the inventory $\inventory$. The dynamics of $\stateVariable$ is controlled by an execution rate $\inventoryRate$ in $\spaceExecutionRates$. In order to emphasise this dependence, we can write $\stateVariable=\stateVariable^{r}$, where $r$ is the control in the space $\spaceExecutionRates$ of execution rates. With this notation, we express the objective function $H=H^{\inventoryRate}$ of the classical stochastic optimisation problem as 
\begin{equation}\label{eq.objectiveFunction}
H^{\inventoryRate} (t,x_1,x_2) :=
\Expectation_{t,\fundamentalPrice_t=x_1, \inventory_t=x_2}
\left[
\int_{t}^{\timeHorizon} F(s,X^{\inventoryRate}_s,\inventoryRate_s) ds
\right],
\end{equation}
where $\inventory$ is in $\spaceInventoryTrajectoriesFuel^{0,\initialInventory}$, and where $F=F(t,X,r)=F(t,\fundamentalPrice,\inventory,r)$ is a Lagrangian that describes risk-adjusted execution-impacted costs from trade. The stochastic optimisation problem for fuel-constrained inventory trajectories is therefore written as 
\begin{equation}\label{eq.fuelConstrainedStochOptProblem}
\inf \left\lbrace H^{\inventoryRate} (0,\fundamentalPrice_0,\initialInventory): \, \inventory \in \spaceInventoryTrajectoriesFuel^{0,\initialInventory} \right\rbrace . 
\end{equation}

An important aspect in the definition of the Lagrangian $F$ in equation \eqref{eq.objectiveFunction} is the description of how the trade execution impacts the price, i.e. the market impact. In this work we focus on the so-called temporary market impact.    

Let $\fundamentalPrice_t$ denote the price process at time $t$. We say that the liquidator exerts a temporary market impact on $\fundamentalPrice_t$ if for some function $g$ in $C(\R\squared)$ the execution price of her order at time $t$ is 
\begin{equation*}
g(\fundamentalPrice_t,\inventoryRate_t),
\end{equation*} 
where $t\mapsto\inventory_t$ is the liquidator's inventory trajectory, and $\inventoryRate_{t}$ denotes its time derivative at time $t$. A well-known example of temporary market impact is given by $g(S,r) = S+\coeffMarketImpact\squared r$, for some coefficient $\coeffMarketImpact>0$ of market impact. In this case, the execution cost is a linear function of the rate of execution $\inventoryRate$; since in a liquidation $\inventory$ is decreasing, the steepest the inventory trajectory is at time $t$, the smaller the execution price is at time $t$. The classical formulation in \cite{AC01opt} utilises this linear temporary market impact. 

In the following two paragraphs \ref{sec.reductionStatic} and \ref{sec.errorsOfLiquidation}, we introduce the concepts of reduction to static optimal strategies and the concept of liquidation error. We show that, in the context of linear temporary market impact with quadratic inventory cost, fuel-constrained optimal liquidation strategies are bound to be static, and non-fuel constrained optimal liquidation strategies commit biased errors of liquidation. This motivates the search for a formulation of the problem of optimal execution that is alternative to the classical one of equation \eqref{eq.fuelConstrainedStochOptProblem}. A possible alternative will then be presented in Section \ref{sec.goodTradeExecutions}; under this alternative, optimal liquidation strategies will be non-static, and -- despite being non-fuel constrained -- they will have unbiased liquidation errors.

\subsection{Reduction to static optimal strategies}\label{sec.reductionStatic}
In a temporary market impact model, trade revenues gained in the infinitesimal time $dt$  are  $
-g(\fundamentalPrice_t, \inventoryRate_t) \inventoryRate_t dt$. When the temporary market impact is linear, this becomes
\begin{equation*}
\left(-\fundamentalPrice_t \inventoryRate_t - \coeffMarketImpact\squared \inventoryRate_t\squared \right) dt,
\end{equation*}
where revenues decompose in a first summand $\fundamentalPrice_t \inventoryRate_t$ where the price process appears, and a second summand $\coeffMarketImpact\squared \inventoryRate_t\squared$ that does not comprise the price process. Clearly, such a decomposition holds in more general situations than the one of linear market impact. If this decomposition holds for the whole Lagrangian $F$ and if the bounded variation component $\boundedVariationPart$ of the price process $\fundamentalPrice$ is deterministic, then we observe the reduction of optimal dynamic solutions to optimal static ones.   This happens in some cases studied in the literature (see \cite[Sections 6.3 and 6.4]{CJP15alg}), where the optimal inventory trajectory, although sought dynamic, is eventually found to be static. It means that the optimiser of \eqref{eq.fuelConstrainedStochOptProblem} is in the space $\spaceInventoryTrajectoriesStatic^{0,\initialInventory}$ of static inventory trajectories. The following proposition explains this phenomenon, pointing out those aspects of the model that cause the reduction to static trade executions.

\begin{prop}[``Reduction to static optimal trade executions'']\label{prop.reductionToStaticSolution}
Assume that 
\begin{equation}\label{eq.decompositionOfLagrangian}
F(t,X,r) = rS + \Lagrangian(t,q,r),
\end{equation}
for some Caratheodory function\footnote{See Definition \ref{def.caratheodory} in Appendix \ref{sec.eulerLagrangeInPresenceOfPricePath} for the definition of Caratheodory function. 
The function $\Lagrangian$ in the statement of Proposition \ref{prop.reductionToStaticSolution} is assumed to be a Caratheodory function with the choices: 1. the open interval $(0,\timeHorizon)$ as the subset $U$ of $\Rn$ in Definition \ref{def.caratheodory}; 2. the two-dimensional variable $(q,r)$ as the variable  $\xi$ in Definition \ref{def.caratheodory}.
}
$\Lagrangian$ that does not depend on $\fundamentalPrice$. Assume  that there exist an integrable function $\alpha$ on $\timeWindow$  and a constant $\beta \geq 0$ such that 
\begin{equation*}
\abs{\Lagrangian (t,q,r)} \leq \alpha(t) + \beta \left(q\squared + r \squared \right).
\end{equation*}
Let the price process $\fundamentalPrice$ be a totally square integrable continuous canonical semimartingale with canonical decomposition  
\begin{equation}\label{eq.evolutionOfFundamentalPrice}
\fundamentalPrice_t = \fundamentalPrice_0 + \boundedVariationPart_t + \martingale_t, \qquad 0\leq t \leq \timeHorizon.
\end{equation}
Assume that $\boundedVariationPart$ is $\sigmaAlgebra_0$-measurable, namely that the drift of the price process is deterministic. 
Then, for all $0\leq t\leq \timeHorizon$ it holds
\begin{equation*}
\inf \Big\lbrace H^{\inventoryRate} (t,\fundamentalPrice_t,\initialInventoryAtTimeT) : \, \inventory \in \spaceInventoryTrajectoriesStatic^{t,\initialInventoryAtTimeT} \Big\rbrace
= \inf \Big\lbrace 
H^{\inventoryRate} (t,\fundamentalPrice_t,\initialInventoryAtTimeT) : \, \inventory \in \spaceInventoryTrajectoriesFuel^{t,\initialInventoryAtTimeT} \Big\rbrace.
\end{equation*}
\end{prop}

\begin{proof}
We give first the proof in the case where $\martingale$ in the canonical decomposition of $\fundamentalPrice$ is a martingale.

Let $\inventory$ be in $\spaceInventoryTrajectoriesFuel^{t,\initialInventoryAtTimeT}$. Let $\stateVariable$ be the state variable  $\stateVariable = (\fundamentalPrice,\inventory)$ and let $Y$ be the two dimensional path $Y=(\inventory,\boundedVariationPart)$. Let $\varphi$ in $C^{\infty}(\R^2)$ be the function $\varphi(x,y)=xy$. Notice that 
\begin{equation*}
\varphi(\stateVariable^{\inventoryRate}_{r} ) - \varphi(\stateVariable^{\inventoryRate}_{t})
-\int_{t}^{r}\stateVariable^{\inventoryRate}_s dY_s, \qquad t\leq r,
\end{equation*}
is a centred martingale. Hence,
\begin{equation*}
\begin{split}
H^{\inventoryRate} (t,\fundamentalPrice_t,\initialInventoryAtTimeT) = &
\Expectation_t \Big[
\int_{t}^{\timeHorizon} F(s,\stateVariable^{\inventoryRate}_s, \inventoryRate_s) ds \\
& \qquad + \varphi (\stateVariable^{\inventoryRate}_\timeHorizon)
- \varphi (\stateVariable^{\inventoryRate}_t)
-\int_{t}^{\timeHorizon} \stateVariable^{\inventoryRate}_r dY_r 
\Big] \\
=& - \initialInventoryAtTimeT \fundamentalPrice_t 
+\Expectation_t \Big[
\int_{t}^{\timeHorizon} \left(-\inventory_r d\boundedVariationPart_s + \Lagrangian(r,\inventory_r,\inventoryRate_r)dr\right).
\Big]
\end{split}
\end{equation*}
It holds
\begin{equation*}
\begin{split}
\inf_{\inventory\in \spaceInventoryTrajectoriesFuel^{t,\initialInventoryAtTimeT}}
H^{\inventoryRate} &(t,\fundamentalPrice_t,\initialInventoryAtTimeT)\\
\geq& 
-\initialInventoryAtTimeT \fundamentalPrice_t
+ \Expectation_t \Big[
\inf_{\inventory \in \spaceInventoryTrajectoriesPathwise^{t,\initialInventoryAtTimeT}}
\int_{t}^{\timeHorizon} \left(-\inventory_r d\boundedVariationPart_s + \Lagrangian(s,\inventory_s,\inventoryRate_s)ds\right)
\Big],
\end{split}
\end{equation*}
where the infimum on the right hand side is taken in a pathwise sense for each realisation of the price $\fundamentalPrice$. In fact, the integrand does not depend on such a realisation (i.e. it does not depend on $\omega$ in $\Omega$) because $\boundedVariationPart$ is non-random. Therefore, any minimising sequence for the infimum inside the expectation is actually independent of $\omega$ and we have 
\begin{equation*}
\begin{split}
\inf_{\inventory\in \spaceInventoryTrajectoriesFuel^{t,\initialInventoryAtTimeT}}
H^{\inventoryRate} &(t,\fundamentalPrice_t,\initialInventoryAtTimeT)\\
\geq &
-\initialInventoryAtTimeT \fundamentalPrice_t
+ \inf_{\inventory \in \spaceInventoryTrajectoriesStatic^{t,\initialInventoryAtTimeT}}
\Expectation_t \Big[
\int_{t}^{\timeHorizon} \left(-\inventory_r d\boundedVariationPart_s + \Lagrangian(s,\inventory_s,\inventoryRate_s)ds \right)
\Big]\\
=&
\inf_{\inventory \in \spaceInventoryTrajectoriesStatic^{t,\initialInventoryAtTimeT}}
H^{\inventoryRate} (t,\fundamentalPrice_t,\initialInventoryAtTimeT).
\end{split}
\end{equation*}
This yields the stated equality in the case where $\martingale$ is a martingale. If instead $\martingale$ is only a local martingale, a standard localisation argument concludes the proof.  
\end{proof}

\begin{remark}\label{remark.optimalityinAC01opt}
 The classical optimal trade  execution proposed by \cite{AC01opt} was originally formulated  with optimality claimed over the set of static inventory trajectories and under the assumption that the price process is an arithmetic Brownian motion. However, it is easy to show that the same solution of the static optimisation is obtained if the Brownian motion is replaced by any square-integrable martingale. In this sense, the static optimal solution of Almgren and Chriss is robust. In view of Proposition \ref{prop.reductionToStaticSolution}, this robustness actually extends to the case where the liquidation strategy is regarded as the optimiser over  the class of fuel-constrained inventory trajectories.
\end{remark}

\begin{remark}\label{remark.GS11dynamicSolution}
A simple case where the optimal trading strategy is non-static is discussed by \cite{GS11opt}. This means that the optimal inventory trajectory obtained from the stochastic control problem is in the space $\spaceInventoryTrajectoriesFuel^{0,\initialInventory}\setminus \spaceInventoryTrajectoriesStatic^{0,\initialInventory}$.  In view of Proposition \ref{prop.reductionToStaticSolution}, we understand the dynamism of their solution by noticing the following. The risk measure adopted by those authors (see \cite[Section 2.1]{GS11opt}) is the value-at-risk of the position $\inventory_t \fundamentalPrice_t$, and this has the consequence of disrupting the assumption that the Lagrangian $F$ can be decomposed as in equation \eqref{eq.decompositionOfLagrangian}. Indeed, Gatheral and Schied consider the optimisation 
\begin{equation}\label{eq.GS11optimisationProblem}
\inf\left\lbrace
\Expectation \Big[
\int_{0}^{\timeHorizon} \left(\inventoryRate_t\squared + \lambda \inventory_t\fundamentalPrice_t\right) dt
\Big]
 : \, \inventory \in \spaceInventoryTrajectoriesFuel^{0,\initialInventory}
\right\rbrace,
\end{equation}
where the price process  $\fundamentalPrice_t = \exp(\sigma\brownianMotion_t -\sigma\squared t /2)$ is the exponential martingale of $\sigma\brownianMotion$, where $\brownianMotion$ denotes  the standard one-dimensional Brownian motion, and where $\lambda$ and $\sigma$ are positive coefficients. Equation \eqref{eq.GS11optimisationProblem} is  \cite[Equation (2.7)]{GS11opt}.
Alternatively, it can be noticed that the same minimisation as in equation \eqref{eq.GS11optimisationProblem} is produced by choosing $F(t,X,r)=rS +  r\squared$ and $dS=\lambda Sdt +\sigma SdW$. Indeed, the expected cost 
\begin{equation}\label{eq.GS11alternativeFormulation}
\begin{split}
\Expectation \Big[
\int_{0}^{\timeHorizon} \left(\inventoryRate_t\squared +   \inventoryRate_t\fundamentalPrice_t\right) dt
\Big],& \\
& \text{ with } 
d\fundamentalPrice_t = \lambda \fundamentalPrice_t dt + \sigma \fundamentalPrice d\brownianMotion_t,
\end{split}
\end{equation}
differs from the expected cost in equation  \eqref{eq.GS11optimisationProblem} (where the price process is the exponential martingale) only by a constant. With the modelling choices in equation \eqref{eq.GS11alternativeFormulation}, the Lagrangian $F$ does not incorporate any risk criterion and thus it satisfies the assumptions of Proposition \ref{prop.reductionToStaticSolution}, but the price process $\fundamentalPrice$ has a position dependent drift coefficient, violating  the assumption that $\boundedVariationPart$ in equation \eqref{eq.evolutionOfFundamentalPrice} is deterministic.
\end{remark}

\begin{remark}
In view of Proposition \ref{prop.reductionToStaticSolution}, we understand why incorporating signals (i.e. short-term price predictors) in the framework of optimal trade execution leads to dynamic optimal strategies (see \cite{CJ16inc} and  \cite{LN19inc}). Indeed, signals are incorporated by modelling the price evolution as \[
d\fundamentalPrice_t = I_t dt + \sigma(t,\fundamentalPrice_t)dW_t,
\]
where $I_t$ is a Markov process that represents the signal. The stochasticity of $I$ disrupts the assumption on the drift $\boundedVariationPart$ in Proposition \ref{prop.reductionToStaticSolution}.
\end{remark}

\begin{corol}
\label{corol.reductionToOptimalStatic}
Assume the setting of Proposition \ref{prop.reductionToStaticSolution}. Assume that the  price  process is modelled as the diffusion 
\begin{equation}\label{eq.SDEofFundamentalPrice}
	d\fundamentalPrice_t = \mu(t)dt +\sigma(t,\fundamentalPrice_t)dW_t,
\end{equation}
for some measurable Lipschitz coefficients $\mu$ and $\sigma$. Assume that the drift coefficient $\mu$ is a deterministic function of time only.
Assume that  
\begin{enumerate}
\item for all $t$ the map $(\inventory, r)\mapsto -\mu(t)\inventory + \Lagrangian(t,\inventory,r)$ is strictly convex;
\item there exist exponents $p>m\geq 1$ and coefficients $\alpha_1>0$, $\alpha_2,\alpha_3 \in \R$ such that 
\begin{equation*}
-\mu(t)\inventory + \Lagrangian(t,\inventory, r) \geq \alpha_1 \abs{r}^{p} + \alpha_2 \abs{\inventory}^m + \alpha_3,
\end{equation*}
for all $t$, $\inventory$ and $r$. 
\end{enumerate} 
Then, the infimum in equation \eqref{eq.fuelConstrainedStochOptProblem} is attained for some optimal deterministic $\inventory$ in $\spaceInventoryTrajectoriesStatic^{0,\initialInventory} \cap W^{1,p}(0,\timeHorizon)$, where $W^{1,p}(0,\timeHorizon)$ denotes the Sobolev space of absolutely continuous function such that their $p$-th power and the $p$-th power of their derivative are integrable on the time interval $(0,\timeHorizon)$. 
\end{corol}

\begin{remark}\label{remark.reductionToStaticSolutionInCarteaTextbook}
The assumptions on $F$ in Corollary \ref{corol.reductionToOptimalStatic} are satisfied in particular by the classical choice 
\begin{equation}\label{eq.classicalLagrangianWithLinearTemporaryImpact}
F(t,\fundamentalPrice, \inventory,r)
=r\fundamentalPrice + \coeffMarketImpact\squared  r\squared + \coeffRiskAversion\squared \, \inventory\squared ,
\end{equation}
where $\coeffMarketImpact>0$ is a coefficient of temporary market impact and $\coeffRiskAversion\geq 0 $ is a coefficient of risk aversion (or of inventory cost). Therefore, Corollary \ref{corol.reductionToOptimalStatic} explains why in \cite[Section 6.3]{CJP15alg} the optimal solution is sought dynamic and eventually found to be static. This also says that, although in \cite{AC01opt} the optimal trade execution was sought only over the class $\spaceInventoryTrajectoriesStatic$ for tractability, this was in fact without loss of generality (Remark \ref{remark.optimalityinAC01opt}).
\end{remark}

\begin{proof}[Proof of Corollary \ref{corol.reductionToOptimalStatic}]
The fact that the infimum over $\spaceInventoryTrajectoriesFuel^{0,\initialInventory}$ is actually the same as the infimum over $\spaceInventoryTrajectoriesStatic^{0,\initialInventory}$ follows from Proposition \ref{prop.reductionToStaticSolution}.  Existence, uniqueness and $p$-integrability of the minimiser follow from the two assumptions on the function $(t,\inventory,r) \mapsto -\mu(t)\inventory +\Lagrangian(t,\inventory,r)$; see \cite[Theorem 4.1]{Dac08dir}.
\end{proof}

\subsection{Errors of liquidation}\label{sec.errorsOfLiquidation}
The space $\spaceInventoryTrajectoriesFuel^{0,\initialInventory}$ of fuel-constrained admissible inventory trajectories has been isolated from the space $\spaceInventoryTrajectories^{0,\initialInventory}$ of first integrals of execution rates. An inventory trajectory $\inventory$ in $\spaceInventoryTrajectories^{0,\initialInventory}\setminus \spaceInventoryTrajectoriesFuel^{0,\initialInventory}$ is said to commit a liquidation error, because with positive probability $\inventory_\timeHorizon \neq \liquidationTarget$. Liquidation errors are common among dynamic solutions to optimal trade execution problems. This is because the mathematical techniques used for dynamic solutions are not well-suited to simultaneously impose the two constraints $\inventory_0 = \initialInventory$ and $\inventory_\timeHorizon = \liquidationTarget$. Clearly, the constraint $\inventory_0 = \initialInventory$ has the priority and hence the constraint $\inventory_\timeHorizon = \liquidationTarget$ is relaxed. The usual relaxation entails to introduce a terminal penalisation for the outstanding inventory at final time. Hence, if $F$ is the Lagrangian describing risk-adjusted cost of trade, it is custom to relax the minimisation in equation \eqref{eq.fuelConstrainedStochOptProblem} and consider instead the problem
\begin{equation}
\label{eq.introOfTerminalPenalisation}
\inf\left\lbrace
\Expectation \Big[
\int_{0}^{\timeHorizon} F(t,\stateVariable_t^{\inventoryRate},\inventoryRate_{t}) dt + \coeffPenalisationOutstandingInventory\squared \Big(\inventory_\timeHorizon - \liquidationTarget\Big)\squared
\Big] : \, 
\inventory \in \spaceInventoryTrajectories^{0,\initialInventory}
\right\rbrace,
\end{equation}
where $\coeffPenalisationOutstandingInventory\geq 0$ is a coefficient of penalisation for outstanding terminal inventory. Notice that the minimisation is performed over the broad class $\spaceInventoryTrajectories^{0,\initialInventory}$ of first integrals of execution rates. Notice also that the objective function in equation \eqref{eq.introOfTerminalPenalisation} can be expressed in the general form discussed so far because 
\begin{equation*}
\begin{split}
\Expectation \Big[
\int_{0}^{\timeHorizon} F(t,\stateVariable_t^{\inventoryRate},\inventoryRate_{t}) &dt + \coeffPenalisationOutstandingInventory\squared \Big(\inventory_\timeHorizon - \liquidationTarget\Big)\squared
\Big] 
=
\Expectation \Big[  \int_{0}^{\timeHorizon} G(t,\stateVariable_t^{\inventoryRate},\inventoryRate_{t}) dt \Big],
\end{split}
\end{equation*}
where $ G(t,\stateVariable_t^{\inventoryRate},\inventoryRate_{t}) =  F(t,\stateVariable_t^{\inventoryRate},\inventoryRate_{t}) + 2\coeffPenalisationOutstandingInventory\squared \inventoryRate_{t} \inventory_t$. 

We isolate liquidation errors whose expected value is null from those that on average either finish the liquidation before the time horizon $\timeHorizon$ (negative liquidation error) or after it (positive liquidation error).
\begin{defi}
  We say that the admissible inventory trajectory $\inventory$ in $\spaceInventoryTrajectories^{0,\initialInventory}$ has an unbiased liquidation error if $\Expectation[\inventory_\timeHorizon] = \liquidationTarget$. We say that $\inventory$ has a biased liquidation error if instead $\Expectation[\inventory_\timeHorizon] \neq \liquidationTarget$.
\end{defi}

By extension we say that a liquidation strategy is unbiased if its inventory trajectory has unbiased liquidation error, and we say that it is biased if it is not unbiased. 

The next proposition shows that the classical optimisation problem corresponding to linear temporary market impact with quadratic inventory cost produces in general optimal liquidation strategies with biased liquidation error. In other words, if the optimal inventory trajectory in this framework happens to have unbiased liquidation error, this unbiasedness is not robust with respect to the values of the model parameters $\coeffMarketImpact$, $\coeffRiskAversion$ and $\coeffPenalisationOutstandingInventory$: independently changing these values will disrupt the expected value of the inventory at $\timeHorizon$, turning it into a biased termination.

The solution to the optimisation problem is derived from \cite{BMO18opt}. Notice that the statement is general with respect to the distributional assumption of the price process, which is only assumed to be a totally square integrable semimartingale.

\begin{prop}
  Let the price process $\priceProcess$ be a totally square integrable special semimartingale. Consider the minimisation problem 
\begin{equation}
\label{eq.BMO18optimisation}
\inf\left\lbrace
\Expectation\Big[
\int_{0}^{\timeHorizon}F(t,\fundamentalPrice_t,\inventory_t,\inventoryRate_{t}) dt 
+ \coeffPenalisationOutstandingInventory\squared \Big(\inventory_\timeHorizon - \liquidationTarget\Big)\squared 
- \inventory_{\timeHorizon} \fundamentalPrice_\timeHorizon 
\Big] : \, 
\inventory \in \spaceInventoryTrajectories^{0,\initialInventory}
\right\rbrace,
\end{equation}
where the Lagrangian $F$ is $F(t,\fundamentalPrice_t,\inventory_t,\inventoryRate_{t})  = \inventoryRate_{t} \fundamentalPrice_t + \coeffMarketImpact\squared \inventoryRate_{t}\squared + \coeffRiskAversion\squared \inventory_{t}\squared$.
Let $\hat{\inventory}$ be the minimiser for \eqref{eq.BMO18optimisation}, and let $M:= \lbrace (\coeffMarketImpact, \coeffRiskAversion, \coeffPenalisationOutstandingInventory) \in \R_{+}^{3}: \,\, \Expectation[\hat{\inventory}_T]=\liquidationTarget\rbrace$. Then, $M$ is included in a manifold of dimension $2$.
\end{prop}
\begin{proof}
Let $\ratioAversionOverImpact$ be the ratio $\ratioAversionOverImpact=\coeffRiskAversion/\coeffMarketImpact$ between the coefficient $\coeffRiskAversion$ of risk aversion and the coefficient $\coeffMarketImpact$ of linear temporary market impact. Let $\ratioTerminalPenalisationOverImpact$ be the ratio $\ratioTerminalPenalisationOverImpact=\coeffPenalisationOutstandingInventory/\coeffMarketImpact$ between the coefficient $\coeffPenalisationOutstandingInventory$ of penalisation of outstanding inventory at time $\timeHorizon$  and the coefficient $\coeffMarketImpact$ of linear temporary market impact. Define the functions $\varphi$ and $\Phi$ as follows:
\begin{equation}\label{eq.timeFunctionsInBMO18solution}
\begin{split}
\varphi(t) :=& \ratioAversionOverImpact \cosh (\ratioAversionOverImpact t) + \ratioTerminalPenalisationOverImpact\squared  \sinh(\ratioAversionOverImpact t), \qquad t\geq 0, \\
\Phi(s,t) : = & \frac{\varphi(\timeHorizon - t)}{\varphi(\timeHorizon - s)} , \qquad 0\leq s \leq t \leq \timeHorizon.
\end{split}
\end{equation}
Let $v(t)$ be the following conditional expectation at time $t$:
\begin{equation}
\label{eq.conditionalExpectationInBMO18solution}
\begin{split}
v(t) :=& \Expectation_t \Big[\frac{1}{2\coeffMarketImpact\squared}\int_{t}^{\timeHorizon} \Phi(t,r)d\fundamentalPrice_r \Big]
=
\Expectation_t \Big[\frac{1}{2\coeffMarketImpact\squared}\int_{t}^{\timeHorizon} \Phi(t,r)d\boundedVariationPart_r \Big],
\end{split}
\end{equation}
where $\fundamentalPrice$ is the price process with canonical decomposition $\fundamentalPrice_t  = \fundamentalPrice_0 + \boundedVariationPart_t + \martingale_t$. 
Then, \cite[Theorem 3.1]{BMO18opt} proves that the optimal inventory trajectory that solves the minimisation problem in equation \eqref{eq.BMO18optimisation} is 
\begin{equation}\label{eq.BMO18solution}
  \hat{\inventory}_t = \Phi(0,t)\initialInventory + \intzerot \Phi(s,t)v(s) ds.
\end{equation}
This minimiser produces unbiased liquidation errors only if 
\begin{equation*}
\frac{\varphi(0)}{\varphi(\timeHorizon)}\initialInventory 
= \frac{1}{2\coeffMarketImpact\squared} \int_{0}^{\timeHorizon} \Phi(t,\timeHorizon)\Expectation \Big[\int_{t}^{\timeHorizon} \Phi(t,r)d\boundedVariationPart_r\Big]dt.
\end{equation*}
Consider $\varphi$ and $\Phi$ as functions of $(\coeffMarketImpact, \coeffRiskAversion, \coeffPenalisationOutstandingInventory)$. Let $f: \R_{+}^{3}\rightarrow \R$ be defined as 
\begin{equation*}
 f(\coeffMarketImpact, \coeffRiskAversion, \coeffPenalisationOutstandingInventory)
 =
 \frac{\varphi(0)}{\varphi(\timeHorizon)}\initialInventory 
- \frac{1}{2\coeffMarketImpact\squared} \int_{0}^{\timeHorizon} \Phi(t,\timeHorizon)\Expectation \Big[\int_{t}^{\timeHorizon} \Phi(t,r)d\boundedVariationPart_r\Big]dt.
\end{equation*}
Then, $0$ is a regular value of $f$ and $M\subset f\inverse(0)$. 
\end{proof}
\begin{remark}
  In the spirit of Proposition \ref{prop.reductionToStaticSolution}, we remark that the solution $\hat{\inventory}$ to the minimisation problem in equation \eqref{eq.BMO18optimisation} is static if the drift of the price process is deterministic, in particular if the price process is a martingale. 
\end{remark}
\begin{remark}
  When the price process is a martingale, the optimal inventory trajectory of equation \eqref{eq.BMO18solution} is such that the terminal value is 
\begin{equation*}
\inventory_\timeHorizon = \frac{\ratioAversionOverImpact \initialInventory}{\ratioAversionOverImpact \cosh (\ratioAversionOverImpact \timeHorizon ) + \ratioTerminalPenalisationOverImpact\squared \sinh(\ratioAversionOverImpact \timeHorizon)}.
\end{equation*}
In this case then, the optimal inventory trajectory will always finish with a positive inventory left to liquidate after the initially fixed time horizon $\timeHorizon$ of the liquidation. 
\end{remark}

\begin{remark}
A liquidation strategy that is unbiased for any choice of $\coeffMarketImpact$ and $\coeffRiskAversion$ is obtained from equation \eqref{eq.BMO18solution} only in the limit as $\coeffPenalisationOutstandingInventory \uparrow \infty$, which yields the fuel-constrained solution 
\begin{equation*}
\begin{split}
\inventory_t =& R(0,t)\initialInventory - \frac{1}{2\coeffMarketImpact\squared} \int_{0}^{t} R(s,t)\fundamentalPrice_s ds \\ 
& - \frac{1}{2\coeffMarketImpact\squared} \int_{0}^{t} R(s,t) \int_{s}^{\timeHorizon} \Expectation_s \left[ \fundamentalPrice_r\right] \partial_r R(s,r) dr ds,
\end{split}
\end{equation*}
where $R(s,t)= \sinh(\ratioAversionOverImpact (\timeHorizon - t)) / \sinh(\ratioAversionOverImpact (\timeHorizon - s))$. This says that the inventory trajectory in equation \eqref{eq.BMO18solution} has unbiased liquidation error only in the degenerate case of deterministic terminal inventory. 
\end{remark}

\section{Good trade executions} \label{sec.goodTradeExecutions}

Let $\fundamentalPrice$ be a price process as defined in Definition \ref{def.priceprocess}. 
Let the class $\spaceExecutionRates$ of inventory rates and the class $\spaceInventoryTrajectories^{0,\initialInventory}$ of admissible inventory trajectories be as defined in equations \eqref{eq.definitionOfSpaceExecutionRates} and \eqref{eq.definitionOfSpaceInventoryTrajectories} respectively. 


From the class $\spaceInventoryTrajectories^{0,
\initialInventory}$ of admissible inventory trajectories we isolate the class of unbiased admissible inventory trajectories. 
An unbiased  admissible inventory trajectory is defined as an $(\sigmaAlgebra_t)_t$-adapted process with absolutely continuous paths, with deterministic initial value $\initialInventory$, expected terminal value $\liquidationTarget$, and such that its derivative is in $\spaceExecutionRates$.  More precisely, we define the space $\spaceUnbiasedInventoryTrajectoriesInitialConstraint$ of unbiased admissible inventory trajectories as 
\begin{equation*}
\spaceUnbiasedInventoryTrajectoriesInitialConstraint:= \left\lbrace
 (\inventory_t)_{t\in\timeWindow} \in \spaceInventoryTrajectories^{0,\initialInventory}: \quad 
 \Expectation [\inventory_\timeHorizon] = \liquidationTarget
\right\rbrace.
\end{equation*}
The constraint $\Expectation [\inventory_\timeHorizon] = \liquidationTarget$ relaxes the fuel constraint $ \inventory_\timeHorizon = \liquidationTarget$ used in the definition of $\spaceInventoryTrajectoriesFuel^{0,\initialInventory}$. Recall that, without loss of generality, the liquidation target $\liquidationTarget$ is set equal to $0$; nonetheless, we do not suppress it from our equations because this makes the formulae easier to interpret (see Remarks \ref{remark.turnGoodExecutionIntoStaticAndIntoAposteriori} and \ref{remark.varianceOfLiquidationError}). 

We consider the following minimisation problem over the class  of unbiased admissible inventory trajectories: 
\begin{equation}\label{eq.optimisationGoodExecution}
\inf \left\lbrace
\intZeroTimeHorizon F (t,\fundamentalPrice_t, \inventory_t,\inventoryRate_{t} )dt: \,  \inventory \in \spaceUnbiasedInventoryTrajectoriesInitialConstraint
\right\rbrace,
\end{equation}
where $F=F(t,x_1,x_2,x_3): (0,\timeHorizon)\times \R^{3} \rightarrow \R$ is a space-differentiable Caratheodory function.\footnote{See Definitions \ref{def.caratheodory} and \ref{def.space-diff_caratheodory}.} 
We use the symbol $\costFunctional$ to denote the map $\inventory  \mapsto \intZeroTimeHorizon F (t,\fundamentalPrice_t, \inventory_t,\inventoryRate_{t} )dt$, for $\inventory$ in $\spaceUnbiasedInventoryTrajectoriesInitialConstraint$. 

\begin{assumption}\label{assumption.Lagrangian}
 Let $F=F(t,x_1,x_2,x_3): (0,\timeHorizon)\times \R^{3} \rightarrow \R$ be the Lagrangian in the minimisation problem \eqref{eq.optimisationGoodExecution}. It is assumed that $F$ is a space-differentiable Caratheodory function, and that the function  $\Lagrangian=\Lagrangian(t,x_1,x_2,x_3):=F(t,x_1,x_2,x_3)-x_1 x_3$ is such that: 1. $\Lagrangian$ is in the Sobolev space $\sobolevSpace[1,4] ((0,\timeHorizon)\times K)$ for all compact subsets $K$ of $\R^{3}$; 2. for almost every $t$ in $(0,\timeHorizon)$, $(\partial_{x_2}\Lagrangian)\squared(t,0,x,0)+(\partial_{x_3}\Lagrangian)\squared(t,0,0,x)=0$ only if $x=0$; 3. the functions $t\mapsto 1/\Lagrangian(t,0,1,0)$ and $t\mapsto 1/\Lagrangian(t,0,0,1)$ are non-negative and square-integrable over  $(0,\timeHorizon)$.
\end{assumption}

Assumption \ref{assumption.Lagrangian} is used to associate the Lagrangian $F$ in equation \eqref{eq.optimisationGoodExecution} with a weight function on the Sobolev space $\sobolevSpaceOneTwo(0,\timeHorizon)$ and with a weight function on $\spaceUnbiasedInventoryTrajectoriesInitialConstraint$.

\begin{defi}\label{def.pwFnorm}
 Let $F=F(t,x_1,x_2,x_3): (0,\timeHorizon)\times \R^{3} \rightarrow \R$ be a space-differentiable Caratheodory function. Let $F$ satisfy Assumption \ref{assumption.Lagrangian}. Let $\eta$ be in $\sobolevSpaceOneTwo(0,\timeHorizon)$. Then, the pathwise $F$-weight $\lvert \eta\rvert_{F}$ of $\eta$ is defined by the equation 
 \begin{equation}\label{eq.definitionPathWiseFnorm}
  \lvert \eta \rvert\squared_{F}
  =
  \int_{0}^{\timeHorizon}
  \left(
  \ell_{2}\squared(t,\eta_t)
  + \ell_{3}\squared(t,\dot{\eta}_t)
  \right)
  dt,
 \end{equation}
where 
\begin{equation*}
 \begin{split}
  \ell_2(t,x)
  :=
  \frac{\partial_{x_2}\Lagrangian(t,0,x,0)}{2\sqrt{\Lagrangian(t,0,1,0)}}, 
\qquad
\ell_3(t,x)
  :=
  \frac{\partial_{x_3}\Lagrangian(t,0,0,x)}{2\sqrt{\Lagrangian(t,0,0,1)}}.  
 \end{split}
\end{equation*}
\end{defi}

\begin{remark}\label{remark.degenerateAssumption}
 Let $F:(0,\timeHorizon)\times\R^{3}\rightarrow\R$ be a space-differentiable Caratheodory function, and define $\Lagrangian(t,x):=F(t,x)-x_1 x_3$. Assume that $\Lagrangian$ satisfies points 1. and 2. in Assumption \ref{assumption.Lagrangian}. Assume that $t\mapsto 1/\Lagrangian(t,0,0,1)$ is non-negative and square-integrable. If $\partial_{x_2}\Lagrangian(t,0,x,0)=0$ for all $t$ and all $x$, then we drop the requirement that $t\mapsto 1/\Lagrangian(t,0,1,0)$ is non-negative and square-integrable and we understand equation \eqref{eq.definitionPathWiseFnorm} with the convention that $\ell_2 \equiv 0$.  
\end{remark}

\begin{defi}\label{def.Fnorm}
 Let $F$ be as in Definition \ref{def.pwFnorm}. Let $\eta$ be in $\spaceUnbiasedInventoryTrajectoriesInitialConstraint$. Then, the $F$-weight $\norm[\eta]_{F}$ of $\eta$ is defined by the following equation
 \begin{equation}\label{eq.definitionFnorm}
  \norm[\eta]_{F}\squared
  =
  \Expectation \left[ \lvert \eta \rvert_{F}\squared\right],
 \end{equation}
 where $\lvert \eta \rvert_{F}$ is the random variable  $\omega\mapsto \lvert \eta(\omega) \rvert_{F}$, where $\lbrace \eta(\omega): \, \omega \in \Omega\rbrace$ are the paths of $\eta$ and, for every $\omega$ in $\Omega$, $\lvert \eta(\omega) \rvert_{F}$ is the pathwise $F$-weight of $t\mapsto \eta_t (\omega)$. 
\end{defi}

Every square-integrable random variable $\xi$ with $\Expectation[\xi] = \liquidationTarget$ identifies a subclass of trajectories in $\spaceUnbiasedInventoryTrajectoriesInitialConstraint$ with specified terminal (random) variable. More precisely, for every $\xi$ in $\Ltwo\probabilitySpace$ with  $\Expectation[\xi] = \liquidationTarget$  we define 
\begin{equation*}
\spaceUnbiasedInventoryTrajectoriesInitialConstraint(\xi) := \left\lbrace
\inventory \in \spaceUnbiasedInventoryTrajectoriesInitialConstraint : \, \Prob(\inventory_\timeHorizon=\xi)=1 
\right\rbrace.
\end{equation*}
The class $\spaceUnbiasedInventoryTrajectoriesInitialConstraint(\xi)$ with $\xi \equiv \liquidationTarget$ is the class of inventory trajectories that commit no liquidation error, namely $\spaceUnbiasedInventoryTrajectoriesInitialConstraint(\liquidationTarget) = \spaceInventoryTrajectoriesFuel^{0,\initialInventory}$. 

\begin{defi}[``Optimal execution of terminal variable $\xi$'']\label{defi.optimalExecutionOfSpecifiedTerminalVariable}
Let $\xi$ be in $\Ltwo \probabilitySpace$ with $\Expectation[\xi] = \liquidationTarget$. We say that $\inventory$ in $\spaceUnbiasedInventoryTrajectoriesInitialConstraint$ is the optimal execution of terminal variable $\xi$ if $\inventory$ minimises $\eta \mapsto \costFunctional(\eta)$ over $\spaceUnbiasedInventoryTrajectoriesInitialConstraint(\xi)$, namely if $ \Prob(\inventory_\timeHorizon =\xi)=1 $ and for all $\eta \in \spaceUnbiasedInventoryTrajectoriesInitialConstraint(\xi)$ it holds
\begin{equation*}
\begin{split}
\intZeroTimeHorizon F(t,\fundamentalPrice_t,\eta_t,&\dotEta_t )dt 
\geq
\intZeroTimeHorizon F(t,\fundamentalPrice_t, \inventory_t ,\inventoryRate_t )dt
\end{split}
\end{equation*}
with probability one.
\end{defi}
For every $\inventory \in \spaceUnbiasedInventoryTrajectoriesInitialConstraint$ we trivially have that $\inventory \in \spaceUnbiasedInventoryTrajectoriesInitialConstraint(\inventory_\timeHorizon)$. We can define a ``tubular'' neighbourhood of $\spaceUnbiasedInventoryTrajectoriesInitialConstraint(\inventory_\timeHorizon) $ by looking at those trajectories $\eta$ in $\spaceUnbiasedInventoryTrajectoriesInitialConstraint$ such that the $\Ltwo(\Prob)$-norm of the difference $\eta_\timeHorizon - \inventory_\timeHorizon$ between terminal values is controlled by the $F$-weight of the difference  $\eta - \inventory$. More precisely, for $\inventory$ in $\spaceUnbiasedInventoryTrajectoriesInitialConstraint$ and $C\geq 0$ we set 
\begin{equation}\label{eq.staticTubularNeighbourhoodIC}
\spaceUnbiasedInventoryTrajectoriesInitialConstraint(\inventory,C):=
\left\lbrace
\eta \in \spaceUnbiasedInventoryTrajectoriesInitialConstraint: \, 
\norm[\eta_\timeHorizon-\inventory_\timeHorizon ]_{\Ltwo(\Prob)} \leq 
C \Fnorm{\eta -\inventory}\squared 
\right\rbrace.
\end{equation}
This captures the idea of $\eta_\timeHorizon$ not being too far from the terminal value $\inventory_\timeHorizon$ given that the trajectory $\eta$  has kept close to $\inventory$ in the time window $0\leq t<\timeHorizon$.

Also, we define a pathwise analogous to the tubular neighbourhood of equation \eqref{eq.staticTubularNeighbourhoodIC}. Given a non-negative $\xi$ in $\Ltwo(\Prob)$ we define
\begin{equation}\label{eq.pathwiseTuburalNeighbourhoodIC}
\spaceUnbiasedInventoryTrajectoriesInitialConstraint_{\text{pw}}(\inventory,\xi):=
\Big\lbrace
\eta \in \spaceUnbiasedInventoryTrajectoriesInitialConstraint : \, 
\abs{\eta_\timeHorizon-\inventory_\timeHorizon} 
\leq \xi \lvert \eta-\inventory \rvert_{F}\squared
\Big\rbrace,
\end{equation}
where $\abs{\eta_\timeHorizon-\inventory_\timeHorizon}$ is the absolute value of the difference between the values of $\eta$ and of $\inventory$ at time $\timeHorizon$, and $\lvert \eta-\inventory \rvert_{F}$ is the pathwise $F$-weight of the difference $\eta - \inventory$.
\begin{remark}\label{remark.tubular_neigh_depend_on_F}
 Notice that both $\spaceUnbiasedInventoryTrajectoriesInitialConstraint(\inventory,C)$ and $\spaceUnbiasedInventoryTrajectoriesInitialConstraint_{\text{pw}}(\inventory,\xi)$ depend on the Lagrangian $F$. Nonetheless, we omit this dependence from the notation, and the symbols for these tubular neighbourhoods do not carry reference to $F$. 
\end{remark}

\begin{defi}[``Good trade execution'']\label{defi.goodTradeExecution}
We say that $\inventory$ in $\spaceUnbiasedInventoryTrajectoriesInitialConstraint$ is a $(C,\xi)$-good trade execution for the minimisation in equation \eqref{eq.optimisationGoodExecution} if there exist $\xi \in \Ltwo_{+} (\Prob)$ and $C\geq 0$ such that 
\begin{enumerate}
\item for all $\eta$ in $\spaceUnbiasedInventoryTrajectoriesInitialConstraint(\inventory,C)$ it holds
\begin{equation*}
\begin{split}
		 \Expectation &  \left[	\int_{0}^{T}   F(t,\fundamentalPrice_t, \eta_t , \dotEta_t )dt \right]
		 \geq
			  \Expectation\left[ \int_{0}^{T}  F(t,\fundamentalPrice_t, \inventory_t,\inventoryRate_t)dt\right];
\end{split}
\end{equation*}
\item for all $\eta$ in $\spaceUnbiasedInventoryTrajectoriesInitialConstraint_{\text{pw}}(\inventory,\xi)$ it holds
\begin{equation*}
\begin{split}
\int_{0}^{T}F(t,&\fundamentalPrice_t,\eta_t ,\dotEta_t )dt 
\geq
 		\int_{0}^{T} F(t,\fundamentalPrice_t, \inventory_t,\inventoryRate_t)dt,
\end{split}
\end{equation*}
with probability one.
\end{enumerate}
\end{defi}
When we emphasise the path $t\mapsto\inventory_t$ of a good trade execution, we use  interchangeably the term \emph{good inventory trajectory}.

\begin{remark}
A good trade execution is in particular an optimal execution of its own terminal variable: if $\inventory$ is as in Definition \ref{defi.goodTradeExecution}, then $\inventory$ is an optimal execution of terminal variable $\inventory_\timeHorizon$ as defined in Definition \ref{defi.optimalExecutionOfSpecifiedTerminalVariable}. In other words, a $(C,\xi)$-good trade execution is in particular a $(0,0)$-good trade execution.
\end{remark}

\subsection{Quadratic inventory cost}\label{sec.IC}

In this section (Section \ref{sec.IC}), we consider the following Lagrangian $F$: 
\begin{equation}\label{eq.LagrangianIC}
F(t,\fundamentalPrice,q,r)
:=
r\fundamentalPrice + \coeffMarketImpact\squared r\squared +\coeffRiskAversion\squared q\squared,
\end{equation}
where $\coeffMarketImpact>0$ is a coefficient of market impact, and $\coeffRiskAversion\geq 0$ is a coefficient of risk aversion. For future reference, we set $\ratioAversionOverImpact:=\coeffRiskAversion/\coeffMarketImpact$. 
Notice that in fact $F$ does not depend on $t$.  We study the problem in \eqref{eq.optimisationGoodExecution} with $F$ as in equation \eqref{eq.LagrangianIC}.

\begin{remark}
The Lagrangian $F$ in equation \eqref{eq.LagrangianIC} represents risk-adjusted revenues from trade where the market impact is temporary and linear, and the risk criterion is quadratic inventory cost.  This aligns to common modelling choices such as those in \cite{LN19inc} and in \cite{BMO18opt}.
However, our relaxation of the fuel constraint entails that the inventory is sought in $\spaceUnbiasedInventoryTrajectoriesInitialConstraint$: we do not modify the objective function as is instead common in the studies of optimal dynamic liquidation strategies, where the terms of terminal asset position $\inventory_\timeHorizon \fundamentalPrice_\timeHorizon$ and of terminal inventory cost $\coeffPenalisationOutstandingInventory \squared (\inventory_\timeHorizon - \liquidationTarget)\squared$ are usually added to the function that describes revenues from trade (see beginning of Section \ref{sec.errorsOfLiquidation}). 
\end{remark}

\begin{remark}
 The Lagrangian $F$ in equation \eqref{eq.LagrangianIC} is the same as the Lagrangian in equation \eqref{eq.classicalLagrangianWithLinearTemporaryImpact}. However, the optimisation in equation \eqref{eq.optimisationGoodExecution} is pathwise and hence it  differs from the classical optimisation of expected risk-adjusted revenues used in equation \eqref{eq.fuelConstrainedStochOptProblem}. For this reason, the martingale cancellation exploited in the proof of Proposition \ref{prop.reductionToStaticSolution} is not applicable to the  present case: we will be able to produce a non-static solution also in the case where the price process has deterministic drift (in particular, where the price process is a martingale). 
\end{remark}

\begin{lemma}
Let $F$ be as in equation \eqref{eq.LagrangianIC}. Then, $F$ satisfies Assumption \ref{assumption.Lagrangian}. Moreover, the pathwise $F$-weight is  a seminorm on $\sobolevSpaceOneTwo(0,\timeHorizon)$, and the $F$-weight is a seminorm on $\spaceUnbiasedInventoryTrajectoriesInitialConstraint$. If $\coeffRiskAversion>0$, then these seminorms are norms. 
\end{lemma}
\begin{proof}
As for the requirements in Assumption \ref{assumption.Lagrangian}, we only notice that the case $\coeffRiskAversion=0$ is covered in Remark \ref{remark.degenerateAssumption}.
As for the second part of the claim, we apply Lemma \ref{lemma.SobolevSeminorm} from Appendix \ref{sec.eulerLagrangeInPresenceOfPricePath} to see that the pathwise $F$-weight is a seminorm.  The fact that the $F$-weight is a seminorm follows from the fact that the pathwise $F$-weight is a seminorm. Finally, $\coeffRiskAversion>0$ guarantees that $\abs{\eta-\tilde{\eta}}_{F} >0$ if $\eta\neq\tilde{\eta}$.
\end{proof}
We denote the pathwise seminorm induced by the pathwise $F$-weight by $\lvert\cdot\rvert_{\coeffMarketImpact, \coeffRiskAversion}$. More precisely, we set 
\begin{equation}\label{eq.pathwiseSobolevNormIC}
\lvert \eta \rvert_{\coeffMarketImpact,\coeffRiskAversion}\squared := \intZeroTimeHorizon \left( \coeffRiskAversion\squared \eta_t \squared + \coeffMarketImpact\squared \dotEta_t\squared \right)dt,
\end{equation}
for $\eta$ in $\sobolevSpaceOneTwo(0,\timeHorizon)$. 
Moreover, we denote the seminorm on $\spaceUnbiasedInventoryTrajectoriesInitialConstraint$ induced by the $F$-weight  by $\norm[\cdot]_{\coeffMarketImpact, \coeffRiskAversion}$.

\subsubsection{Closed-form formula}\label{sec.closedformIC}

\begin{prop} \label{prop.goodExecutionIC}
Let $F$ be as in equation \eqref{eq.LagrangianIC}. Let $\ratioAversionOverImpact$ be the ratio of the coefficients of risk aversion and of market impact, namely $\ratioAversionOverImpact := \coeffRiskAversion/\coeffMarketImpact$. 
Let $\alpha$ be the function $\alpha(t)=1-\sinh(\ratioAversionOverImpact(T-t))/\sinh(\ratioAversionOverImpact T)$,
and let $K$ be the constant 
\[
K= \frac{1}{2\coeffMarketImpact\squared \sinh(\ratioAversionOverImpact T)}
\int_{0}^{T} \cosh(\ratioAversionOverImpact (T-u))\Expectation\big[\fundamentalPrice_u\big]du. 
\]
For $0\leq t \leq \timeHorizon$, define 
\begin{equation}\label{eq.goodExecutionIC}
\begin{split}
\inventory_t:= & 
\big(1-\alpha(t)\big) \initialInventory + \alpha(t)\liquidationTarget \\
&- \frac{1}{2\coeffMarketImpact\squared}
\int_{0}^{t} \cosh\big(\ratioAversionOverImpact (t-u)\big)\fundamentalPrice_udu \\
&+ K \sinh(\ratioAversionOverImpact t).
\end{split}
\end{equation}
Then, $(\inventory_t)_{t\in\timeWindow}$ is a $(C,\xi)$-good trade execution. The constant $C$ is explicitly given by the formula
\[
C\inverse = \ratioAversionOverImpact
\int_{0}^{T} \sinh\big(\ratioAversionOverImpact (T-u)\big)\Variance^{\half}\big(\fundamentalPrice_u\big) du;
\]
the random variable $\xi$ is explicitly given by the formula
\begin{equation*}
\begin{split}
\xi\inverse = 
\Big\lvert & 
2\coeffMarketImpact\coeffRiskAversion \frac{ \liquidationTarget - \initialInventory }{\sinh(\ratioAversionOverImpact \timeHorizon)} 
-\ratioAversionOverImpact \intZeroTimeHorizon \sinh (\ratioAversionOverImpact (\timeHorizon - t))\fundamentalPrice_t dt \\
 & +\ratioAversionOverImpact \frac{\cosh(\ratioAversionOverImpact \timeHorizon)}{\sinh(\ratioAversionOverImpact\timeHorizon)}
 \intZeroTimeHorizon \cosh(\ratioAversionOverImpact (\timeHorizon - t))\Expectation[\fundamentalPrice_t] dt 
 \Big\rvert. 
\end{split}
\end{equation*}
\end{prop}
\begin{remark}\label{remark.turnGoodExecutionIntoStaticAndIntoAposteriori}
The structure of the solution $\inventory$ in equation \eqref{eq.goodExecutionIC} is threefold: a time-dependent convex combination between initial inventory $\initialInventory$ and liquidation target $\liquidationTarget$ appears on the first line; a dynamic response to the actual price trajectory appears on the second line; an adjustment for the terminal constraint $\Expectation[\inventory_\timeHorizon]  = \liquidationTarget$ appears on the third line. 

If in the integral appearing on the second line of equation  \eqref{eq.goodExecutionIC} we replace the fundamental price $\fundamentalPrice_u$ with its expected value $\Expectation[\fundamentalPrice_u]$, then  the inventory trajectory $\inventory$ is turned into the optimal static one, i.e. into the minimiser of $\Expectation [\costFunctional(\eta)]$ over all $\eta$ in $\spaceInventoryTrajectoriesStatic^{0,\initialInventory}$. See Corollary \ref{corol.staticOptimalSolIC} in Appendix \ref{sec.eulerLagrangeInPresenceOfPricePath}. 

Instead, if in the definition of the constant  $K$ appearing on  the third line of equation  \eqref{eq.goodExecutionIC} we replace $\Expectation[\fundamentalPrice_u]$ with $\fundamentalPrice_u$, then the inventory trajectory $\inventory$ is turned into the optimal a-posteriori one, i.e. into the minimiser of $\costFunctional(\eta)$ over all $\eta$ in $\spaceInventoryTrajectoriesPathwise^{0,\initialInventory}$. This is an immediate application of Proposition \ref{prop.weakFormEulerLagrange} in Appendix \ref{sec.eulerLagrangeInPresenceOfPricePath}.
\end{remark}
\begin{proof}[Proof of Proposition \ref{prop.goodExecutionIC}]
Let $\inventory$ be as in equation  \eqref{eq.goodExecutionIC}. The fact that $\inventory$ is in $\spaceUnbiasedInventoryTrajectoriesInitialConstraint$ is apparent. Let $f_t:= 2\coeffMarketImpact\squared \inventoryRate_{t} + \fundamentalPrice_t$ and notice that $f$ is absolutely continuous with derivative 
\begin{equation}\label{eq.eulerLagrangeWithCancellationIC}
\dot{f}_t = 2 \coeffMarketImpact\squared \inventory_t .
\end{equation}
Let $\eta$ be in $\spaceUnbiasedInventoryTrajectoriesInitialConstraint$. We write $e$ for the difference $e:=\eta-\inventory$, and we observe that $e_0=0$ and $\Expectation e_\timeHorizon = 0$.  Then, we have 
\begin{equation*}
\begin{split}
\costFunctional(\eta) - \costFunctional(\inventory) = &
\intZeroTimeHorizon \Big[ f_t \dot{e}_t + 2\coeffMarketImpact\squared \inventory_t \, e_t\Big]dt
+\int_{0}^{\timeHorizon}\big(\coeffMarketImpact\squared \dot{e}\squared + \coeffRiskAversion\squared e\squared \big)dt.
\end{split}
\end{equation*}
The second integral on the right hand side is $ \lvert e \rvert_{\coeffMarketImpact,\coeffRiskAversion}\squared$.  
Using integration-by-parts we see that in fact $\costFunctional(\eta) - \costFunctional(\inventory) $ $= f_\timeHorizon e_\timeHorizon $ $ + \lvert e \rvert_{\coeffMarketImpact,\coeffRiskAversion}\squared$, because of equation \eqref{eq.eulerLagrangeWithCancellationIC}.  
Therefore, the difference $\costFunctional(\eta) - \costFunctional(\inventory) $ is non-negative if 
\begin{equation*}
\abs{e_\timeHorizon}
\leq \frac{ \lvert e \rvert_{\coeffMarketImpact,\coeffRiskAversion}\squared}{\abs{f_\timeHorizon}}.
\end{equation*}
This gives $\xi = 1/\lvert f_\timeHorizon\rvert$.

Secondly, consider the   expected difference $\Expectation \costFunctional (\eta) - \Expectation \costFunctional (\inventory)$ $=\Expectation[f_\timeHorizon e_\timeHorizon] + \normInventoryTrjectories{e}\squared$. We can estimate
\begin{equation*}
\Expectation[f_\timeHorizon e_\timeHorizon] \leq \Variance^{\half}(f_\timeHorizon)\norm[e_\timeHorizon]_{\Ltwo(\Prob)},
\end{equation*}
because $\Expectation e_\timeHorizon = 0$. Moreover, 
\begin{equation*}
\Variance^{\half}(f_\timeHorizon)
\leq \ratioAversionOverImpact \intZeroTimeHorizon \sinh(\ratioAversionOverImpact(\timeHorizon - t))\Variance^{\half}(\fundamentalPrice_t) dt. 
\end{equation*} 
Therefore, the   expected difference $\Expectation \costFunctional (\eta) - \Expectation \costFunctional (\inventory)$ is non-negative if 
\begin{equation*}
\norm[e_\timeHorizon]_{\Ltwo(\Prob)} \leq 
\frac{ \normInventoryTrjectories{e}\squared }{\ratioAversionOverImpact \intZeroTimeHorizon \sinh(\ratioAversionOverImpact(\timeHorizon - t))\Variance^{\half}(\fundamentalPrice_t) dt}.
\end{equation*}
This gives the constant $C$ in the statement and concludes the proof.
\end{proof}

\begin{remark}
We remark that the good inventory trajectory of equation \eqref{eq.goodExecutionIC} is written without assuming a particular SDE dynamics for the price evolution. In particular, Proposition \ref{prop.goodExecutionIC} applies to the case in which the price process is modelled as a fractional Brownian motion, or as the sum of a possibly discontinuous semimartingale and a fractional Brownian motion. 
Moreover, the good inventory trajectory is robust, in the sense that it retains its optimality when one price process $\priceProcess$ is replaced by another price process $\tilde{\priceProcess}$ with $\Expectation[\priceProcess_t] = \Expectation[\tilde{\priceProcess}_t]$ for all $t$.  
\end{remark}

\begin{remark}\label{remark.varianceOfLiquidationError}
Given  $\inventory$ as in equation \eqref{eq.goodExecutionIC}, we can compute 
\begin{equation*}
\norm[\inventory_\timeHorizon-\liquidationTarget]_{\Ltwo(\Prob)}
= \frac{1}{2\coeffMarketImpact\squared} \Expectation^{\half}
\left[\big(
\int_{0}^{\timeHorizon} \cosh\big(\ratioAversionOverImpact (T-t)\big) \Big(
\fundamentalPrice_t-\Expectation[\fundamentalPrice_t]
\Big)dt
\big)\squared\right]
\end{equation*}
and estimate
\begin{equation*}
\Variance(\inventory_\timeHorizon)
\leq \frac{T}{4\coeffMarketImpact^{4}} 
\int_{0}^{T} \cosh\squared\big(\ratioAversionOverImpact (T-t)\big) \Variance(\fundamentalPrice_t)dt.
\end{equation*}
We therefore remark the following two facts. First, the smaller $\int \Variance(\fundamentalPrice_t)dt$ is, the more precise the good execution $\inventory$ of Proposition \ref{prop.goodExecutionIC} is. Second,  the square of the coefficient $\coeffMarketImpact$ of linear market impact is inversely proportional to the standard deviation of $\inventory_{\timeHorizon}$, and thus the precision with which the good execution $\inventory$ of equation \eqref{eq.goodExecutionIC} gets to its liquidation target $\liquidationTarget$ increases when the strategy itself can exert more influence on the execution price.
\end{remark}

\begin{remark}	\label{remark.alternativeFinalConstraints}
Unbiased admissible inventory trajectories $\inventory$ have been defined as absolutely continuous stochastic processes on $\timeWindow$ such that $\Expectation[\inventory_\timeHorizon] = \liquidationTarget$. This has meant that the constant $K$ in equation \eqref{eq.goodExecutionIC} has been chosen to minimise $\Expectation [(\inventory_\timeHorizon - \liquidationTarget) \squared]$. We can give two alternatives to this minimisation:
\begin{enumerate}
	\item Choose $K$ in such a way to minimise 
	\begin{equation*}
	\Expectation \left[ \fint_{t_0}^{\timeHorizon} (\inventory_t-\liquidationTarget)\squared dt \right] ,
	\end{equation*}
	for some $0\leq t_0 <\timeHorizon$. The symbol $ \fint_{t_0}^{\timeHorizon} dt$ stands for the mean $ \frac{1}{\timeHorizon - t_0}\int_{t_0}^{\timeHorizon}dt$. This yields
	\begin{equation*}
	K= \frac{1}{2\coeffMarketImpact\squared \fint_{t_0}^{\timeHorizon}\sinh\squared (\ratioAversionOverImpact t) dt}
	\fint_{t_0}^{\timeHorizon} \sinh(\ratioAversionOverImpact t) \psi(t) dt,
	\end{equation*}
	where $\psi(t)=\intzerot \cosh(\ratioAversionOverImpact (t-u))\Expectation[\fundamentalPrice_u] du $ $- (1-\alpha(t))(\initialInventory - \liquidationTarget)$. Notice that this tends to the former choice when $t_0 \uparrow\timeHorizon$.
	\item Choose $K$ in such a way to minimise
	\begin{equation*}
	\Expectation \left[
	\Big(
	\fint_{t_0}^{\timeHorizon}\inventory_tdt -\liquidationTarget
	\Big)\squared
	\right],
	\end{equation*}
	for some $0\leq t_0 <\timeHorizon$. This yields
	\begin{equation*}
	K= \frac{\ratioAversionOverImpact (\timeHorizon - t_0)}
	{2\coeffMarketImpact\squared 
		\left(
		\cosh(\ratioAversionOverImpact\timeHorizon) - \cosh(\ratioAversionOverImpact t_0)
		\right)
	}
	\fint_{t_0}^{\timeHorizon} \psi(t) dt,
	\end{equation*}
	with $\psi$ as above. 
\end{enumerate}
Notice however that the alternative choices for the constant $K$ make the corresponding inventory trajectory fall out of the set $\spaceUnbiasedInventoryTrajectories^{0,\initialInventory}$.
\end{remark}

\subsubsection{Characterisation via Euler-Lagrange equation}\label{sec.eulerLagrangeIC}
The differential equation in \eqref{eq.eulerLagrangeWithCancellationIC} is the linchpin on which the derivation of the good inventory trajectory is based. This equation is the Euler-Lagrange equation associated with the functional $\costFunctional$. Equation \eqref{eq.eulerLagrangeWithCancellationIC} is a random ordinary differential equation, where differentiation is possible because the price process is cancelled out in the sum $ 2\coeffMarketImpact\squared \inventoryRate_{t} + \fundamentalPrice_t$. Such a cancellation allows to circumvent the need of an integration with respect to the price process. However this integration is possible; Appendix \ref{sec.eulerLagrangeInPresenceOfPricePath} presents the theory of the Euler-Lagrange equation in the presence of a (rough) price path. The motivation for this theory comes from the fact that equation \eqref{eq.eulerLagrangeWithCancellationIC} can be rewritten as 
\begin{equation}\label{eq.eulerLagrangeSystemIC}
\begin{cases}
d\inventory_t=& r_t dt \\
dr_t=& \ratioAversionOverImpact\squared (\inventory_t - \liquidationTarget)dt - d\fundamentalPrice_t/2\coeffMarketImpact\squared.
\end{cases}
\end{equation}
We interpret the system in equation \eqref{eq.eulerLagrangeSystemIC} as a random Young differential equation. This equation  is useful in simulations because it avoids the computation of the integrals in equation \eqref{eq.goodExecutionIC} and the evaluation of hyperbolic functions. Moreover, since we use Young integration to integrate with respect to the price process $\fundamentalPrice$, the system in equation \eqref{eq.eulerLagrangeSystemIC} has a pathwise meaning and thus it also makes sense in the practical implementation of the trading strategy, where the price process $\fundamentalPrice$ is replaced by the single price path observed during the liquidation. 

In this  paragraph, we apply the theory of Appendix \ref{sec.eulerLagrangeInPresenceOfPricePath} in order to characterise the good trade execution of Proposition \ref{prop.goodExecutionIC} in terms of an initial value problem for the dynamics in \eqref{eq.eulerLagrangeSystemIC}.

We start by noticing that the Lagrangian $F$ in equation \eqref{eq.LagrangianIC} satisfies Assumption \ref{assumption.decompositionOfF} in Appendix \ref{sec.eulerLagrangeInPresenceOfPricePath}. Equation \eqref{eq.eulerLagrangeSystemIC} is equation \eqref{eq.strongFormEulerLagrangeEquation} with $F$ given by \eqref{eq.LagrangianIC}. Moreover, by  casting Definition \ref{defi.solutionToSecondOrderRDE} to the case of equation \eqref{eq.eulerLagrangeSystemIC}, we have:
\begin{defi}\label{defi.meaningOfSolutionToEulerLagrangeSystemIC}
Let $\inventory$ be in $\spaceUnbiasedInventoryTrajectoriesInitialConstraint$. We say that $\inventory$ solves equation \eqref{eq.eulerLagrangeSystemIC} if for all $\omega$ in $\Omega$, all $\eta$ in $\smoothCompactlySupportedFunctions (0,\timeHorizon)$ and all $0\leq s\leq t \leq \timeHorizon$ the following holds:
\begin{equation}
\label{eq.meaningOfSolutionToEulerLagrangeSystemIC}
\begin{split}
\int_{s}^{t} \dotEta_u d\inventory_u (\omega) 
=& \eta_t \inventoryRate_{t}(\omega) - \eta_s \inventoryRate_s (\omega) \\
& - \ratioAversionOverImpact\squared 	\int_{s}^{t} \eta_u \big(\inventory_u (\omega) - \liquidationTarget\big) du 
+ \frac{1}{2\coeffMarketImpact\squared} 	\int_{s}^{t} \eta_{u} d\fundamentalPrice_u (\omega).
\end{split}
\end{equation}
\end{defi}
We remark that Definition \ref{defi.meaningOfSolutionToEulerLagrangeSystemIC} is pathwise: a scenario $\omega$ in $\Omega$ could be fixed and the definition would still make sense. The integral on the left hand side of equation \eqref{eq.meaningOfSolutionToEulerLagrangeSystemIC} is well defined because for all $\omega$ in $\Omega$ the path $\inventoryRate(\omega)$ is in $\Ltwo \timeWindow$. Similarly, the first integral on the right hand side has a pathwise meaning and it is well defined because $\inventory(\omega)$ is in $\Ltwo \timeWindow$ for all  $\omega$ in $\Omega$. Thirdly, the integral $\int \eta d\fundamentalPrice$ on the right hand side of equation \eqref{eq.meaningOfSolutionToEulerLagrangeSystemIC} is the Young integral introduced in Lemma \ref{lemma.integrationByPartsYoungIntegral}.

\begin{lemma}
\label{lemma.twoSolutionsEulerLagrangeIC}
Let $\inventory$ and $\tilde{\inventory}$ be two solutions to equation \eqref{eq.eulerLagrangeSystemIC}. If $\ratioAversionOverImpact \neq 0 $, then there exist constants $K_1$ and $K_2$ such that 
\begin{equation*}
\inventory_t - \tilde{\inventory}_t = K_1 e^{\ratioAversionOverImpact t } + K_2 e^{-\ratioAversionOverImpact t};
\end{equation*}
if $\ratioAversionOverImpact=0$, then there exist constants $K_1$ and $K_2$ such that 
\begin{equation*}
\inventory_t - \tilde{\inventory}_t = K_1  + K_2 t.
\end{equation*}
In particular, in both cases the difference between $\inventory$  and $\tilde{\inventory}$ is deterministic. 
\end{lemma}
\begin{proof}
Let $\eta$ be arbitrary in $\smoothCompactlySupportedFunctions(0,\timeHorizon)$. By Definition \ref{defi.meaningOfSolutionToEulerLagrangeSystemIC} we have that 
\begin{equation*}
\begin{split}
\int_{0}^{t} \dotEta_u d\inventory_u 
=& \eta_t \inventoryRate_{t}  
 - \ratioAversionOverImpact\squared 	\int_{0}^{t} \eta_u \big(\inventory_u  - \liquidationTarget\big) du 
+ \frac{1}{2\coeffMarketImpact\squared} 	\int_{0}^{t} \eta_{u} d\fundamentalPrice_u ; \\
\int_{0}^{t} \dotEta_u d\tilde{\inventory}_u 
=& \eta_t \inventoryRate_{t}  
 - \ratioAversionOverImpact\squared 	\int_{0}^{t} \eta_u \big(\tilde{\inventory}_u  - \liquidationTarget\big) du 
+ \frac{1}{2\coeffMarketImpact\squared} 	\int_{0}^{t} \eta_{u} d\fundamentalPrice_u .
\end{split}
\end{equation*}
Let $\epsilon_t:= \inventory_t - \tilde{\inventory}_t$. Subtract one line from the other and obtain that the function 
\begin{equation*}
t\mapsto \intzerot \dotEta_u \dot{\epsilon}_u du + \ratioAversionOverImpact\squared \intzerot \eta_u \epsilon_u du - \eta_t \dot{\epsilon}_t 
\end{equation*}
is constantly null. The first two summands are differentiable in $t$ and hence the third summand $ \eta_t \dot{\epsilon}_t $ is differentiable too. Since $\eta$ is arbitrary, $\dot{\epsilon}$ is differentiable in $(0,\timeHorizon)$. Differentiating $t\mapsto \eta_t \dot{\epsilon}_t$, we obtain
\begin{equation*}
 \dotEta_t \dot{\epsilon}_t + \ratioAversionOverImpact\squared \eta_t \epsilon_t - \dot{\eta}_t \dot{\epsilon}_t  - \eta_t \ddot{\epsilon}_t = 0.
\end{equation*}
Hence $\ddot{\epsilon}_t = \ratioAversionOverImpact\squared \epsilon_t$, proving the lemma. 
\end{proof}

\begin{lemma}
\label{lemma.uniquenessOfEulerLagrangeIC}
The solution $\inventory$ to equation \eqref{eq.eulerLagrangeSystemIC} with constraint
\begin{equation}\label{eq.initialAndTerminalConstraintForGoodTradeExecution}
\begin{cases}
\inventory_0 = \initialInventory \\
\Expectation \left[ \inventory_\timeHorizon \right] = \liquidationTarget 
\end{cases}
\end{equation}
is unique. 
\end{lemma}
\begin{proof}
Let  $\inventory$ and $\tilde{\inventory}$ be two solutions to equation \eqref{eq.eulerLagrangeSystemIC} satisfying the constraints in \eqref{eq.initialAndTerminalConstraintForGoodTradeExecution}. Assume $\ratioAversionOverImpact \neq 0$. Then, by Lemma \ref{lemma.twoSolutionsEulerLagrangeIC} it must be
\begin{equation*}
\inventory_t - \tilde{\inventory}_t = K_1 e^{\ratioAversionOverImpact t } + K_2 e^{-\ratioAversionOverImpact t},
\end{equation*}
for constants $K_1$ and $K_2$. Therefore 
\begin{equation*}
\begin{cases}
\inventory_0 - \tilde{\inventory}_0 = K_1 + K_2 =0 \\
\Expectation \inventory_\timeHorizon  - \Expectation \tilde{\inventory}_\timeHorizon =
 K_1 e^{\ratioAversionOverImpact\timeHorizon } + K_2 e^{-\ratioAversionOverImpact\timeHorizon } =0
\end{cases}
\end{equation*}
Solving for  $K_1$ and $K_2$ we find $K_1=K_2=0$. 

The case  $\ratioAversionOverImpact = 0$ is analogous. 
\end{proof}
Having established uniqueness of the solution to   equation \eqref{eq.eulerLagrangeSystemIC} with constraint \eqref{eq.initialAndTerminalConstraintForGoodTradeExecution}, we link equation \eqref{eq.eulerLagrangeSystemIC} to the good trade execution of Proposition \ref{prop.goodExecutionIC}. This link is established as an application of Lemma \ref{lemma.characterisationOfSolutionToSecondOrderRDE} from Appendix \ref{sec.eulerLagrangeInPresenceOfPricePath}. 

\begin{lemma}
\label{lemma.characterisationEulerLagrangeIC}
Let $\inventory$ be in $\spaceUnbiasedInventoryTrajectoriesInitialConstraint$. Then, the following are equivalent
\begin{enumerate}
	\item the inventory trajectory $\inventory$ solves the Euler-Lagrange equation \eqref{eq.eulerLagrangeSystemIC};
	\item the function $f_t:= 2\coeffMarketImpact\squared \inventoryRate_t + \fundamentalPrice_t$ is absolutely continuous with derivative 
	\begin{equation*}
	\dot{f}_t = 2 \coeffRiskAversion\squared \inventory_t. 
	\end{equation*}
\end{enumerate}
In particular, the $(C,\xi)$-good trade execution in Proposition \ref{prop.goodExecutionIC} solves the Euler-Lagrange equation \eqref{eq.eulerLagrangeSystemIC}. 
\end{lemma}

We are finally in the position to prove the main result of this  paragraph. 
\begin{prop}[``Characterisation of good trade execution via Euler-Lagrange equation''] \label{prop.eulerLagrangeCharacterisationIC}
The good trade execution for the minimisation of \eqref{eq.optimisationGoodExecution} with the Lagrangian $F$ as in equation \eqref{eq.LagrangianIC} is characterised as the solution to the Euler-Lagrange equation in \eqref{eq.eulerLagrangeSystemIC} with initialisation
\begin{equation}\label{eq.initialisationEulerLagrangeIC}
\begin{cases}
\inventory_0 = \initialInventory \\
r_0 = - \fundamentalPrice_0 / (2\coeffMarketImpact\squared) + {\ratioAversionOverImpact}{\sinh\inverse (\ratioAversionOverImpact \timeHorizon)}
			\Big[
				(\liquidationTarget - \initialInventory) \cosh(\ratioAversionOverImpact\timeHorizon) 
				+\tilde{K}
				\Big],
\end{cases}
\end{equation}
where 
\begin{equation*}
\begin{split}
\tilde{K}=&\frac{\cosh(\ratioAversionOverImpact\timeHorizon)}{2\coeffMarketImpact\squared}
\int_{0}^{\timeHorizon} \cosh(\ratioAversionOverImpact t) \Expectation \big[\fundamentalPrice_t \big] dt 
-\frac{\sinh(\ratioAversionOverImpact\timeHorizon)}{2\coeffMarketImpact\squared}
\int_{0}^{\timeHorizon} \sinh(\ratioAversionOverImpact t) \Expectation \big[\fundamentalPrice_t \big] dt.
\end{split}
\end{equation*}	

In particular, the good trade execution in Proposition \ref{prop.goodExecutionIC} is the only good trade execution for the minimisation \eqref{eq.optimisationGoodExecution} with the Lagrangian $F$ as in equation \eqref{eq.LagrangianIC}.
\end{prop}
\begin{remark}
Proposition \ref{prop.goodExecutionIC} gives a characterisation of the good trade execution in terms of an initial value problem that is easily simulated. This is the practical relevance of the characterisation. We will rely on the initial value problem \eqref{eq.eulerLagrangeSystemIC} with initial conditions \eqref{eq.initialisationEulerLagrangeIC} in our numerical experiments in Section \ref{sec.applications}.
\end{remark}
\begin{proof}
First we examine the following two implications.
\begin{enumerate}
	\item \emph{A good trade execution solves equation \eqref{eq.eulerLagrangeSystemIC}}. Let $\tilde{\inventory}$ be a good trade execution for the minimisation of \eqref{eq.optimisationGoodExecution} with the Lagrangian $F$ as in equation \eqref{eq.LagrangianIC}. Then in particular, for every $\omega$ in $\Omega$ it holds
	\begin{equation*}
	\tilde{\inventory}(\omega) = \argmin \left\lbrace \costFunctional(\eta):\, \eta \in \tilde{\inventory}(\omega) + \sobolevSpaceCompactSupport(0,\timeHorizon) \right\rbrace.
	\end{equation*}
	Therefore, by Proposition \ref{prop.eulerLagrangeNecessity}, we have that $\tilde{\inventory}(\omega)$ solves the equation
	\begin{equation*}
	\begin{cases}
	d\tilde{\inventory}_t(\omega)=&   \tilde{r}_t(\omega)dt \\
	d\tilde{r}_t(\omega)=& \ratioAversionOverImpact\squared \tilde{\inventory}_t(\omega)dt - d\fundamentalPrice_t(\omega)/2\coeffMarketImpact\squared.
	\end{cases}
	\end{equation*}
	\item \emph{A solution to equation \eqref{eq.eulerLagrangeSystemIC} is a good trade execution}. Assume that $\inventory$ solves the Euler-Lagrange equation in \eqref{eq.eulerLagrangeSystemIC}. Then, by Lemma \ref{lemma.characterisationEulerLagrangeIC} we have that $f_t(\omega):= 2\coeffMarketImpact\squared \inventoryRate_t(\omega) + \fundamentalPrice_t(\omega)$ is absolutely continuous in $t$ for all $\omega$ in $\Omega$ and its time derivative  is $\dot{f}_t(\omega) = 2 \coeffRiskAversion\squared \inventory_t(\omega)$. Hence, for  all $\omega$ in $\Omega$ and all $\eta$ in $\inventory(\omega) + \sobolevSpaceCompactSupport(0,\timeHorizon)$ we have 
	\begin{equation*}
	\begin{split}
	\costFunctional(\eta) - \costFunctional(\inventory(\omega)) = &
	\intZeroTimeHorizon \big(2\coeffRiskAversion\squared \inventory_t(\omega) - \dot{f}_t (\omega) \big)e_t dt
	+ \intZeroTimeHorizon \big(\coeffMarketImpact\squared \dot{e}_t\squared + \coeffRiskAversion\squared e_t\squared \big) dt,
	\end{split}
	\end{equation*}
	where $e = \eta-\inventory(\omega)$.  The first summand on the right hand side is null and thus $\costFunctional(\eta)\geq \costFunctional(\inventory(\omega))$. This shows that $\inventory$ is a $(0,0)$-good trade execution. 
\end{enumerate}
In view of these two implications and of Lemma \ref{lemma.uniquenessOfEulerLagrangeIC}, it only remains to show that the initialisations in equation \eqref{eq.initialisationEulerLagrangeIC} are equivalent to the constraints in \eqref{eq.initialAndTerminalConstraintForGoodTradeExecution}.

Consider an equation of the form
\begin{equation*}
dY_t = (AY_t + B)dt + D d\fundamentalPrice_t,
\end{equation*}
where the unknown $Y$ is in $\R\squared$, the matrix $A$ is in $\R^{2\times 2}$ and $B$ and $D$ are two-dimensional real vectors. Equation \eqref{eq.eulerLagrangeSystemIC} is of this form with the choice $Y=(\inventory,r)\transpose$ and
\begin{equation*}
A=
\begin{pmatrix}
0 & 1 \\ 
\ratioAversionOverImpact\squared & 0 
\end{pmatrix}, \qquad
B= 
\begin{pmatrix}
0 \\ -\ratioAversionOverImpact\squared \liquidationTarget
\end{pmatrix}
 = 0, \qquad
D= 
\begin{pmatrix}
0 \\ -\frac{1}{2\coeffMarketImpact\squared}
\end{pmatrix}.
\end{equation*} 
Assume first that $t\mapsto \Expectation[\fundamentalPrice_t]$ is differentiable. Then, if we set $\mu_t=\Expectation[Y_t]$ we have that $\mu$ solves the ordinary differential equation
\begin{equation}\label{eq.ODEforExpectedInventory}
\dot{\mu}_t = A\mu_t + \varphi_t D,
\end{equation}
where $\varphi_t= d\Expectation[\fundamentalPrice_t]/dt$. This ordinary differential equation has a two-dimensional space of solutions; hence we can exploit these two degrees of freedom to adjust for the constraints $\inventory_0 = \mu_0 = \initialInventory$ and $\mu_\timeHorizon  = \liquidationTarget = 0$. 

Define the function $e_1=e_1(t)$ as 
\begin{equation*}
\begin{split}
e_1(t) = & e_1(0) + \Expectation[\fundamentalPrice_t]\sinh(\ratioAversionOverImpact t)/ (2\coeffMarketImpact \coeffRiskAversion) \\
& - \frac{1}{2\coeffMarketImpact\squared } \intzerot \cosh(\ratioAversionOverImpact u) \Expectation[\fundamentalPrice_u] du.
\end{split}
\end{equation*}
Define the function $e_2=e_2(t)$ as 
\begin{equation*}
\begin{split}
e_2(t) = & e_2(0) +  \Big(\fundamentalPrice_0 -\cosh(\ratioAversionOverImpact t)\Expectation[\fundamentalPrice_t]\Big)/ (2\coeffMarketImpact \coeffRiskAversion) \\
& + \frac{1}{2\coeffMarketImpact\squared } \intzerot \sinh(\ratioAversionOverImpact u) \Expectation[\fundamentalPrice_u] du.
\end{split}
\end{equation*}
The general solution to equation \eqref{eq.ODEforExpectedInventory} is
\begin{equation*}
\begin{cases}
\mu^{(1)}_t = e_1(t) \cosh(\ratioAversionOverImpact t ) + e_2(t)\sinh(\ratioAversionOverImpact t) \\
\mu^{(2)}_t = \ratioAversionOverImpact e_1(t) \sinh(\ratioAversionOverImpact t ) + \ratioAversionOverImpact e_2(t)\cosh(\ratioAversionOverImpact t).
\end{cases}
\end{equation*}
 The constraints $\inventory_0 = \mu_0 = \initialInventory$ and $\mu_\timeHorizon  = \liquidationTarget = 0$ impose the choices
 \begin{equation*}
 \begin{cases}
 e_1(0)=\initialInventory \\
 e_2(0)=R+\sinh\inverse (\ratioAversionOverImpact \timeHorizon)\Big[\liquidationTarget - e_1(\timeHorizon)\cosh(\ratioAversionOverImpact \timeHorizon)\Big],
 \end{cases}
 \end{equation*}
 where 
 \begin{equation*}
 R= \frac{\cosh(\ratioAversionOverImpact \timeHorizon)\Expectation[\fundamentalPrice_\timeHorizon] - \fundamentalPrice_0}{2\coeffMarketImpact\coeffRiskAversion}
 -\frac{1}{2\coeffMarketImpact\squared}\intZeroTimeHorizon \sinh(\ratioAversionOverImpact u)\Expectation[\fundamentalPrice_u]du.
 \end{equation*}
 Hence, the constraints $\inventory_0 = \mu_0 = \initialInventory$ and $\mu_\timeHorizon  = \liquidationTarget = 0$ are translated into the initialisation in the statement. 
 
 The case where the map $t\mapsto \Expectation[\fundamentalPrice_t]$ is not differentiable is handled via a standard approximation argument.
\end{proof}

\subsection{Alternative risk criteria}\label{sec.alternativeRiskCriteria}
We present two alternatives to the risk criterion used in Section \ref{sec.IC}. This means that we modify the third summand in the Lagrangian of equation \eqref{eq.LagrangianIC}, and we study the minimisation problem with such a modified Lagrangian.
The first alternative (Section \ref{sec.timeIC}) preserves the same structure but increases the weight of the coefficient $\coeffRiskAversion$ of risk aversion linearly in time.  The second alternative (Section \ref{sec.varInspiredRiskCriterion}) is instead inspired by the value-at-risk for geometric Brownian motion used in \cite{GS11opt}.

\subsubsection{Linearly time-dependent coefficient of risk aversion}\label{sec.timeIC}
The third summand in the Lagrangian $F$ of equation \eqref{eq.LagrangianIC} accounts for the risk aversion. So far this term has been taken constant in the time variable $t$. We now propose a linear $t$-dependence, with higher risk aversion for $t$ closer to the liquidation horizon $\timeHorizon$. More precisely, we consider the Lagrangian
\begin{equation}\label{eq.LagrangianTimeIC}
F(t,S,q,r):= rS + \coeffMarketImpact\squared r\squared +\coeffRiskAversion\squared \, t \,  q\squared,
\end{equation}
where $\coeffMarketImpact>0$ is a coefficient of market impact and $\coeffRiskAversion\geq0$ is a coefficient of risk aversion. For future reference, we set $\ratioAversionOverImpact:=\coeffRiskAversion/\coeffMarketImpact$. We study the minimisation problem in \eqref{eq.optimisationGoodExecution} with $F$ as in equation \eqref{eq.LagrangianTimeIC}.

\begin{lemma}
Let $F$ be as in equation \eqref{eq.LagrangianTimeIC}. Then, $F$ satisfies Assumption \ref{assumption.Lagrangian}. Moreover, the pathwise $F$-weight is  a seminorm on $\sobolevSpaceOneTwo(0,\timeHorizon)$, and the $F$-weight is a seminorm on $\spaceUnbiasedInventoryTrajectoriesInitialConstraint$. If $\coeffRiskAversion>0$, then these seminorms are norms. 
\end{lemma}

We denote the pathwise seminorm induced by the pathwise $F$-weight by $\lvert\cdot\rvert_{\coeffMarketImpact, \coeffRiskAversion\sqrt{t}}$. More precisely, we set 
\begin{equation}\label{eq.pathwiseSobolevNormTimeIC}
\lvert \eta \rvert_{\coeffMarketImpact,\coeffRiskAversion\sqrt{t}}\squared := \intZeroTimeHorizon \left( \coeffRiskAversion\squared \, t \,  \eta_t \squared + \coeffMarketImpact\squared \dotEta_t\squared \right)dt,
\end{equation}
for $\eta$ in $\sobolevSpaceOneTwo(0,\timeHorizon)$. 
Moreover, we denote the seminorm on $\spaceUnbiasedInventoryTrajectoriesInitialConstraint$ induced by the $F$-weight  by $\norm[\cdot]_{\coeffMarketImpact, \coeffRiskAversion\sqrt{t}}$.

We proceed with statements analogous to those  in Sections \ref{sec.closedformIC} and \ref{sec.eulerLagrangeIC}, namely: in Proposition \ref{prop.goodTradeExecutionTimeIC} we give a closed-form formula for a good trade execution in the case of the Lagrangian $F$ of \eqref{eq.LagrangianTimeIC}; then, in Proposition \ref{prop.characterisationEulerLagrangeTimeIC} we show that in fact such a good trade execution is unique, and we characterise it as the solution of a random Young differential equation. 
All the arguments are straightforward adaptations from those presented above, and thus we omit the proofs. 

\begin{prop}\label{prop.goodTradeExecutionTimeIC}
Let the Lagrangian $F$ be as in equation \eqref{eq.LagrangianTimeIC}. Let $\airyFirstFunction$ and $\airySecondFunction$ be the first and the second Airy's functions, namely the two independent solutions to the second order linear ordinary differential equation $u^{\prime \prime}(t) - tu(t) = 0$. Define the functions $\alpha$, $\beta$, $\phi$ as follows
\begin{equation*}
\begin{split}
\alpha(t) =& \airyFirstFunction(\ratioAversionOverImpact^{2/3} t),\\
\beta(t) =& \airySecondFunction(\ratioAversionOverImpact^{2/3} t), \\
\phi(t) =& \frac{1}{2\coeffMarketImpact\squared}\intzerot \alpha^{-2}(s)\int_{0}^{s}\alpha(u)d\fundamentalPrice_u \,ds,
\end{split}
\end{equation*}
where the innermost integral in the definition of $\phi$ is the Young integral introduced in Appendix \ref{sec.eulerLagrangeInPresenceOfPricePath}.
Define the constants $c_A$ and $c_B$ as follows
\begin{equation*}
c_A = \frac{\beta(\timeHorizon)\initialInventory - \alpha(\timeHorizon)\beta(0)\Expectation\phi(\timeHorizon)}
{\alpha(0)\beta(\timeHorizon) - \alpha(\timeHorizon)\beta(0)},
\end{equation*}
	\begin{equation*}
c_B = \frac{ \alpha(\timeHorizon)}
{\alpha(0)\beta(\timeHorizon) - \alpha(\timeHorizon)\beta(0)}\Big(\alpha(0)\Expectation\phi(\timeHorizon) - \initialInventory\Big).
\end{equation*}
For $0\leq t \leq \timeHorizon$, define 
\begin{equation}\label{eq.goodTradeExecutionTimeIC}
\inventory_t := c_A \alpha(t) + c_B\beta(t) - \alpha(t) \phi(t).
\end{equation}
Then, $(\inventory_t)_{t\in\timeWindow}$ is a $(C,\xi)$-good trade execution, where
\begin{equation*}
C=1/
\lVert 
\fundamentalPrice_{\timeHorizon} - 2\coeffMarketImpact\squared \dot{\alpha \phi}(\timeHorizon) 
\rVert_{\Ltwo(\Prob)},
\end{equation*}
\begin{equation*}
\xi=1/
\lvert 
\fundamentalPrice_{\timeHorizon} 
+2\coeffMarketImpact\squared c_A\dot{\alpha}(\timeHorizon)
+2\coeffMarketImpact\squared c_B \dot{\beta}(\timeHorizon)
-2\coeffMarketImpact\squared \dot{\alpha \phi}(\timeHorizon)
\rvert,
\end{equation*}
where the symbol $ \dot{\alpha \phi}$ denotes the time derivative of the product function $t\mapsto\alpha(t)\phi(t)$. 
\end{prop}

The Lagrangian $F$ in equation \eqref{eq.LagrangianTimeIC} satisfies Assumption \ref{assumption.decompositionOfF} from Appendix \ref{sec.eulerLagrangeInPresenceOfPricePath}.  Equation \eqref{eq.strongFormEulerLagrangeEquation} in the present case reads
\begin{equation}\label{eq.eulerLagrangeSystemTimeIC}
\begin{cases}
d\inventory_t=r_{t} dt \\
dr_{t} = \ratioAversionOverImpact\squared \, t \, \inventory_t dt - d\fundamentalPrice_t / 2\coeffMarketImpact\squared .
\end{cases}
\end{equation}
We show that  the good trade execution in Proposition \ref{prop.goodTradeExecutionTimeIC} is characterised as the unique solution to the random Young differential equation \eqref{eq.eulerLagrangeSystemTimeIC}.

\begin{prop}\label{prop.characterisationEulerLagrangeTimeIC}
The good trade execution of Proposition \ref{prop.goodTradeExecutionTimeIC} is characterised as the unique solution to the random Young differential equation \eqref{eq.eulerLagrangeSystemTimeIC} with initialisation
\begin{equation*}
\begin{cases}
\inventory_0 = \initialInventory \\
r_0 = e_A(0) \alpha_0 + e_B (0) \beta_0,
\end{cases}
\end{equation*}
where $\alpha$ and $\beta$ are as in Proposition \ref{prop.goodTradeExecutionTimeIC}, and 
\begin{equation*}
\begin{split}
&
e_A (0) 
=
\frac{ \beta_\timeHorizon \initialInventory - \beta_0 \tilde{K}}{\alpha_0\beta_\timeHorizon - \alpha_\timeHorizon \beta_0}, 
\qquad
e_B (0) 
=
\frac{\alpha_0 \tilde{K} - \alpha_\timeHorizon \initialInventory}{\alpha_0\beta_\timeHorizon - \alpha_\timeHorizon \beta_0},
\\
&
\tilde{K} 
=
\frac{\fundamentalPrice_0}{2\coeffMarketImpact\squared W_0}\left(\alpha_\timeHorizon \beta_0 - \alpha_0 \beta_\timeHorizon \right) 
+ \frac{1}{2\coeffMarketImpact\squared} \intZeroTimeHorizon \Expectation\left[\fundamentalPrice_u\right] 
\frac{d}{d u} \left(\frac{\alpha_\timeHorizon \beta_u - \beta_\timeHorizon \alpha_u}{W_u}\right) du , 
\\
&
W_t 
=
\alpha_t \dot{\beta}_t - \dot{\alpha}_t \beta_t. 
\end{split}
\end{equation*}
In particular, the good trade execution in Proposition \ref{prop.goodTradeExecutionTimeIC} is the only good trade execution for the minimisation \eqref{eq.optimisationGoodExecution} with the Lagrangian $F$ as in equation \eqref{eq.LagrangianTimeIC}.
\end{prop}

\subsubsection{VaR-inspired risk criterion}\label{sec.varInspiredRiskCriterion}

\cite{GS11opt} model the fundamental price $\fundamentalPrice$ as a geometric Brownian motion and they adopt the value-at-risk as measure of risk aversion. This means penalising instantaneous revenues  from trade  by subtracting a term proportional to $\inventory_t \fundamentalPrice_t$ at every time $t$. Inspired by their modelling choices, we now consider the Lagrangian
\begin{equation}\label{eq.LagrangianVaR}
F(t,S,q,r):= rS + \coeffMarketImpact\squared r\squared +\coeffRiskAversion\squared \, q S,
\end{equation}
where $\coeffMarketImpact>0$ is a coefficient of market impact and $\coeffRiskAversion\geq0$ is a coefficient of risk aversion. Notice that in fact $F$ does not depend on $t$. We study the minimisation problem in \eqref{eq.optimisationGoodExecution} with $F$ as in equation \eqref{eq.LagrangianVaR}. 

\begin{lemma}
Let $F$ be as in equation \eqref{eq.LagrangianVaR}. Then, $F$ satisfies Assumption \ref{assumption.Lagrangian}. Moreover, the pathwise $F$-weight is  a seminorm on $\sobolevSpaceOneTwo(0,\timeHorizon)$, and the $F$-weight is a seminorm on $\spaceUnbiasedInventoryTrajectoriesInitialConstraint$.  
\end{lemma}
\begin{proof}
 Assumption \ref{assumption.Lagrangian} is understood as per Remark \ref{remark.degenerateAssumption}, namely we drop the requirement on $t\mapsto 1/\Lagrangian(t,0,1,0)$ because $\partial_{x_3}\Lagrangian(\cdot,0,\cdot,0)\equiv 0$.
\end{proof}

We denote the pathwise seminorm induced by the pathwise $F$-weight by $\lvert\cdot\rvert_{\coeffMarketImpact}$. More precisely, we set 
\begin{equation}\label{eq.pathwiseSobolevNormVaR}
\lvert \eta \rvert_{\coeffMarketImpact}\squared := \intZeroTimeHorizon  \coeffMarketImpact\squared \dotEta_t\squared dt,
\end{equation}
for $\eta$ in $\sobolevSpaceOneTwo(0,\timeHorizon)$. 
Moreover, we denote the seminorm on $\spaceUnbiasedInventoryTrajectoriesInitialConstraint$ induced by the $F$-weight  by $\norm[\cdot]_{\coeffMarketImpact}$.

We proceed with statements analogous to those  in Sections \ref{sec.closedformIC} and \ref{sec.eulerLagrangeIC}, namely: in Proposition \ref{prop.goodTradeExecutionVaR} we give a closed-form formula for a good trade execution in the case of the Lagrangian $F$ in equation  \eqref{eq.LagrangianVaR}; then, in Proposition \ref{prop.characterisationEulerLagrangeVaR} we show that such a good trade execution is unique, and we characterise it as the solution of a random Young differential equation. 
All the arguments are straightforward adaptations from those presented above, and thus we omit the proofs.

\begin{prop}\label{prop.goodTradeExecutionVaR}
	Let $F$ be as in equation \eqref{eq.LagrangianVaR}. Let $K$ be the constant 
	\begin{equation*}
	K=\frac{1}{2\coeffMarketImpact\squared \timeHorizon}\int_{0}^{\timeHorizon} 
	\left(
	\Expectation\left[\fundamentalPrice_s\right] - \coeffRiskAversion\squared \int_{0}^{s}\Expectation\left[\fundamentalPrice_u\right]du 
	 \right) ds.
	\end{equation*}
	For $0\leq t\leq\timeHorizon$, define
	\begin{equation}\label{eq.goodExecutionVaR}
	\begin{split}
	\inventory_t =& \left(1-\frac{t}{\timeHorizon}\right)\initialInventory + \frac{t}{\timeHorizon}\liquidationTarget \\
	& - \frac{1}{2\coeffMarketImpact\squared}\int_{0}^{t} 
	\left(
	\fundamentalPrice_s - \coeffRiskAversion\squared \int_{0}^{s}\fundamentalPrice_u du 
	\right) ds\\
	& +Kt
	\end{split}
	\end{equation}
	Then, $(\inventory_t)_{t\in\timeWindow}$ is a $(C,\xi)$-good trade execution. The random variable $\xi$ is explicitly given by the formula  
	\begin{equation*}
	\xi\inverse  = \left\lvert
	\frac{2\coeffMarketImpact\squared}{\timeHorizon} \big(\liquidationTarget - \initialInventory \big)
	+2\coeffMarketImpact\squared K 
	+\coeffRiskAversion\squared \int_{0}^{\timeHorizon} \fundamentalPrice_t dt 
	\right\rvert
	\end{equation*}
	and the constant $C$ is explicitly given by the formula $C\inverse=\lVert \xi\inverse \rVert_{\Ltwo(\Prob)}$. 
\end{prop}

\begin{remark}
The good trade execution in equation \eqref{eq.goodExecutionVaR} has the same structure of the one in equation \eqref{eq.goodExecutionIC}, namely: a time-dependent convex combination of $\initialInventory$ and $\liquidationTarget$ (first line), a dynamic response to the realisation of the price path (second line), and an adjustment for the constraint $\Expectation[\inventory_\timeHorizon] = \liquidationTarget$ (third line). Moreover, notice that 
\begin{equation*}
\lim_{\coeffRiskAversion\downarrow 0}
\frac{\sinh\left(\frac{\coeffRiskAversion}{\coeffMarketImpact} (\timeHorizon - t)\right)}{\sinh\left(\frac{\coeffRiskAversion}{\coeffMarketImpact} \timeHorizon \right)}
= 1 - 	\lim_{\coeffRiskAversion\downarrow 0}
\frac{\sinh\left(\frac{\coeffRiskAversion}{\coeffMarketImpact}  t\right)}{\sinh\left(\frac{\coeffRiskAversion}{\coeffMarketImpact} \timeHorizon \right)}
= 1- \frac{t}{\timeHorizon},
\end{equation*} 
so that the good trade execution in equation \eqref{eq.goodExecutionVaR} and the good trade execution in equation \eqref{eq.goodExecutionIC} agree when the risk aversion vanishes, i.e. in the limit as $\coeffRiskAversion\downarrow 0$. 
\end{remark}

\begin{remark}
Of the good trade execution in equation \eqref{eq.goodExecutionVaR}, we can compute 
\begin{equation*}
\lVert \inventory_\timeHorizon - \liquidationTarget \rVert_{\Ltwo(\Prob)}
= \frac{1}{2\coeffMarketImpact\squared} 
\Expectation^{\half} \Big[
\big(
\int_{0}^{\timeHorizon}
	\Big\lbrace
		\coeffRiskAversion\squared \intzerot (\fundamentalPrice_u - \Expectation[\fundamentalPrice_u])du 
		-  (\fundamentalPrice_t - \Expectation[\fundamentalPrice_t])
	\Big\rbrace dt 
\big)\squared 
\Big]
\end{equation*}
and estimate 
\begin{equation*}
\Variance(\inventory_\timeHorizon) \leq 
\frac{\timeHorizon}{2\coeffMarketImpact^{4}} \int_{0}^{\timeHorizon}
\left(
\coeffRiskAversion^{4} t \intzerot \Variance(\fundamentalPrice_u)du + \Variance(\fundamentalPrice_t)
\right)dt.
\end{equation*}
Therefore, the two facts presented in Remark \ref{remark.varianceOfLiquidationError} also hold for the good trade execution of Proposition \ref{prop.goodTradeExecutionVaR}.
\end{remark}

The Lagrangian $F$ in equation \eqref{eq.LagrangianVaR}  satisfies Assumption \ref{assumption.decompositionOfF} from Appendix \ref{sec.eulerLagrangeInPresenceOfPricePath}. With this $F$, equation \eqref{eq.strongFormEulerLagrangeEquation} reads
\begin{equation}\label{eq.eulerLagrangeSystemVaR}
\begin{cases}
d\inventory_t=r_{t} dt \\
dr_{t} = \ratioAversionOverImpact\squared \, t \, \fundamentalPrice_t dt/2 - d\fundamentalPrice_t / 2\coeffMarketImpact\squared ,
\end{cases}
\end{equation}
where $\ratioAversionOverImpact:=\coeffRiskAversion/\coeffMarketImpact$.
We now characterise the good trade execution in Proposition \ref{prop.goodTradeExecutionVaR} as the unique solution to the random Young differential equation \eqref{eq.eulerLagrangeSystemVaR}. Observe that equation \eqref{eq.eulerLagrangeSystemVaR} is solved by direct integration and this simplifies several calculations compared to the cases of Sections \ref{sec.eulerLagrangeIC} and \ref{sec.timeIC}. 

\begin{prop}\label{prop.characterisationEulerLagrangeVaR}
The good trade execution of Proposition \ref{prop.goodTradeExecutionVaR} is characterised as the unique solution to the random Young differential equation \eqref{eq.eulerLagrangeSystemVaR} with initialisation
\begin{equation*}
\begin{cases}
\inventory_0 =& \initialInventory \\
r_0 =& \frac{\liquidationTarget - \initialInventory}{\timeHorizon}  +\frac{1}{2\coeffMarketImpact\squared \timeHorizon} \intZeroTimeHorizon \left(\Expectation[\fundamentalPrice_t] - \fundamentalPrice_0 - \coeffRiskAversion\squared \intzerot \Expectation[\fundamentalPrice_u]du\right)dt.
\end{cases}
\end{equation*}
In particular, the good trade execution in Proposition \ref{prop.goodTradeExecutionVaR} is the only good trade execution for the minimisation \eqref{eq.optimisationGoodExecution} with $F$ is as in equation \eqref{eq.LagrangianVaR}.
\end{prop}

\section{Applications}\label{sec.applications}
In this section we given two applications of the trading schedule proposed in Proposition \ref{prop.goodExecutionIC}.

Because of the reliance on the expected trajectory $t\mapsto\Expectation[\fundamentalPrice_t]$ of the fundamental price, the use case of our framework is one where such an expected trajectory can serve as a reliable forecast for the price. This means that an implicit mean-reversion is assumed, and we will be choosing our price processes accordingly.   

\subsection{INTC shares: high-frequency mean-reverting jump diffusion}
We consider the liquidation of a large portfolio of shares;  the whole liquidation happens during intraday trading hours and  on a single limit order book.
We model the fundamental price $\fundamentalPrice$ on the mid-price of the order book, according to the following high-frequency mean-reverting jump-diffusion model:
\begin{equation}\label{eq.jumpdiffusion}
 \fundamentalPrice_t
 =
 \exp
 \left(
 m(t)
 +
 Y(t)
 +
 N(t)
 \right),
\end{equation}
where $m(t)$ is $\sigmaAlgebra_0$-measurable, continuous and of finite variation, $Y(t)$ is the Ornstein-Uhlenbeck process $dY(t) = -\alpha Y((t)dt + \sigma dW(t)$, and $N(t)$ is a compound Poisson process  independent from $Y(t)$ and with i.i.d marks symmetrically distributed around zero. Therefore, the expected price trajectory is $\Expectation[\fundamentalPrice_t]=\exp(m(t))$, and it represents the reversion target. The liquidator acts adopting this as her price forecast during the execution. 

As an example, we calibrate our model to the high-frequency NASDAQ order book data available for the trading of INTC on 22 January 2019. The dataset is provided by LOBSTER (\url{https://lobsterdata.com/}). The calibration procedure is straightforward: we first extrapolate the mean-reversion target from the raw data, we then estimate the Ornstein-Uhlenbeck factor from the difference between the log-price and the logarithm of the mean-reversion target, and we finally calibrate the point process on the data points $\lbrace Y(t_i): \, \vert Y(t_i)-Y(t_{i-1})\exp(-\alpha(t_i - t_{i-1})) \vert > k \sigma \sqrt{1-\exp(-2\alpha(t_{i}-t_{i-1}))/\alpha} \rbrace$ that are displaced beyond $k$ standard deviations from the expected value. In the interest of conciseness,  we omit the details of this calibration procedure.  
The so-calibrated price process is displayed in the upper quadrant of Figure \ref{fig.INTCliquidation}, together with the original data stream. 

The lower quadrant in Figure \ref{fig.INTCliquidation} reports the inventory trajectory and its rate of execution computed as a solution to the Cauchy problem of equation \eqref{eq.eulerLagrangeSystemIC} with initialisation \eqref{eq.initialisationEulerLagrangeIC}. 
We remark two aspects observable from the picture. The first aspect is that the liquidation terminates after the initially decided horizon $T=1.0$ for the liquidation. This is because the realisation of the price path dwells for most of the time window of the liquidation below its expected value. Consequently, the liquidator decreases the rate of liquidation, in order to prevent her own price impact from exacerbating her trading cost even further. The trade-off between a limited exposure to market volatility and a parsimonious rate is resolved in favour of the latter.
The second aspect is the reaction to the price jump that happens around $t=0.21$. The jump is upward and hence favourable to the liquidation. Consequently, the liquidator reacts by increasing the rate of execution to exploit this.

\subsection{5Y government bonds: Brownian bridge}
As a second application, we consider an entirely different time scale from the one of intraday INTC stock price, and in fact different from the time scale at which models of trade executions are usually applied. In this sense, the application is at the boundary between trade execution and dynamic portfolio management. 

We consider the sale of 5Y government bonds motivated by market conditions whereby bond yields are negative. Historically, this was observed in 2016 for German bunds and such a phenomenon reappeared in March 2019. Four examples are reported in the upper quadrant of Figure \ref{fig.bonds}, where we show the prices of 5Y German bunds with maturities October 2019, April 2020, October 2020, and April 2021.

The reason for negative yields was twofold. On the one hand, central banks launched programmes to stimulate the economy by cutting interest rates and by injecting liquidity via the so-called quantitative easing. On the other hand, in volatile economic regimes risk-adversed investors tend to move their capital to safe investments such as government bonds (a phenomenon known as flight-to-quality), and this exerted upward pressure on bond prices.  

We consider an investor who wishes to unwind a long position on a negative yielding bond before incurring into the sure loss at the bond's maturity. The sale is done gradually in time, on the one hand because the investor does not want to suddenly loose all the liquidity associated with the bonds, and on the other hand because the investor is mindful of the market impact that an abrupt sale would incur into.  

In the upper quadrant of Figure \ref{fig.bonds} we observe that, after the period of overprice and as the maturity approaches, the trajectories converge on a downward slope towards the face value, which pins the trajectories at maturity. We base our modelling choices on this observation and, in line with a tradition that dates back to the Seventies (\cite{Boy70sto}), we adopt a Brownian bridge as the price process. More precisely, the price process $\fundamentalPrice$ is modelled as 
\begin{equation}\label{eq.priceProcessBond}
d\fundamentalPrice_t = \frac{V - \fundamentalPrice_t}{\timeHorizon - t} dt + \sigma d\brownianMotion_t, 
\end{equation}  
where $V$ is the face value of the bond, $\brownianMotion$ is a standard one-dimensional Brownian motion, and $\sigma$ is the volatility coefficient. In this model for the price process, the expected price path is the line segment from the decision price $\actualPricePath_0$ at time $t=0$ to the face value at maturity. 

The second quadrant in Figure \ref{fig.bonds} shows good trade executions in this setting. The inventory trajectories are simulated following the dynamics in equation \eqref{eq.eulerLagrangeSystemIC} with initialisation \eqref{eq.initialisationEulerLagrangeIC}.

The way in which the good trade executions react to market scenarios is indicative of the dynamic adjustment of inventory trajectories during the liquidation. In particular, in scenarios \#1 and \#3 the liquidation is faster than the static solution; in this way it exploits favourable market conditions. Notice that this is advantageous especially in scenario \#3, where a faster liquidation means that the liquidator concentrates her sale before the price plunges below its expected trend. On the contrary, in scenario \#4, the good inventory trajectory is less steep than the static solution, because unfavourable market conditions recommend to parsimoniously impact the price.  

\section{Conclusions}\label{sec.conclusions}
In this paper, we examined the mathematical models of optimal trade execution with respect to two properties: non-static trajectories and unbiased liquidation errors. Non-static trajectories are those that react to the actual realisation of the price path during the execution, rather than being based only on assumed distributional properties of this price. Secondly, a liquidation error is said to be unbiased if its expectation is zero, entailing that the expected value of the terminal  inventory coincides with the execution target.

We introduced our proposal for execution strategies, which 
enjoy both properties. In particular, in order to have non-static solutions even when the fundamental price is modelled as a martingale, we considered the minimisation of  trading costs from a \emph{pathwise} perspective, rather than the minimisation of \emph{expected} trading costs. 

We considered three risk criteria. The first criterion is the classical quadratic inventory cost; the second is a time-dependent modification of the first; the third was inspired by the value-at-risk employed in \cite{GS11opt}.  For all of them, we derived explicit closed-form formulae of our inventory trajectories. Furthermore, we characterised them through initial value problems that allow to easily implement our strategies in practice. We demonstrated this through two applications, one on the liquidation of INTC shares, the other on the liquidation of German bunds.  

\section*{Acknowledgements}
We are grateful to Eyal Neuman for discussion and suggestions that helped us improve the paper.
We would also like to thank two anonymous referees for their comments and recommendations. 

\section*{}

\begin{center}
\begin{figure}
	\centering
	\begin{tabular}{c}
	 \includegraphics[width=0.89\textwidth]{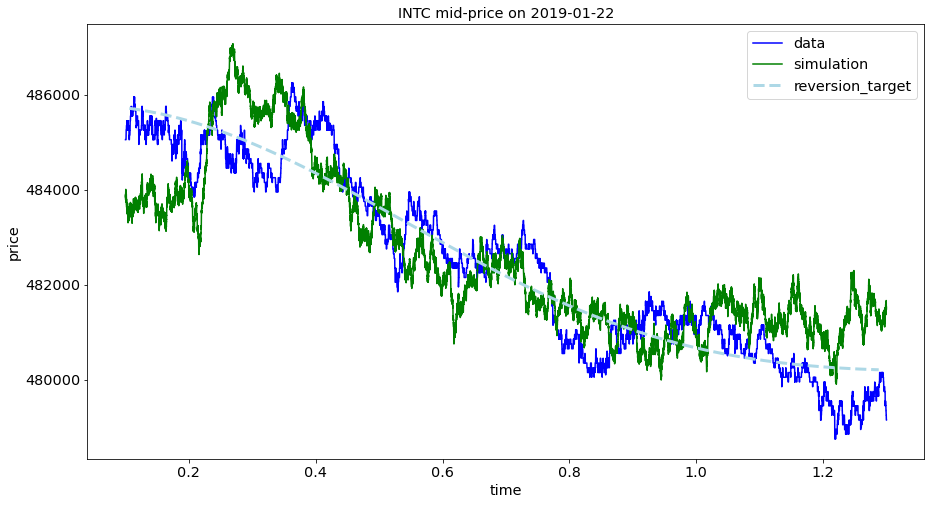}
	 \\
	 \hline
	 \\
	\includegraphics[width=0.89\textwidth]{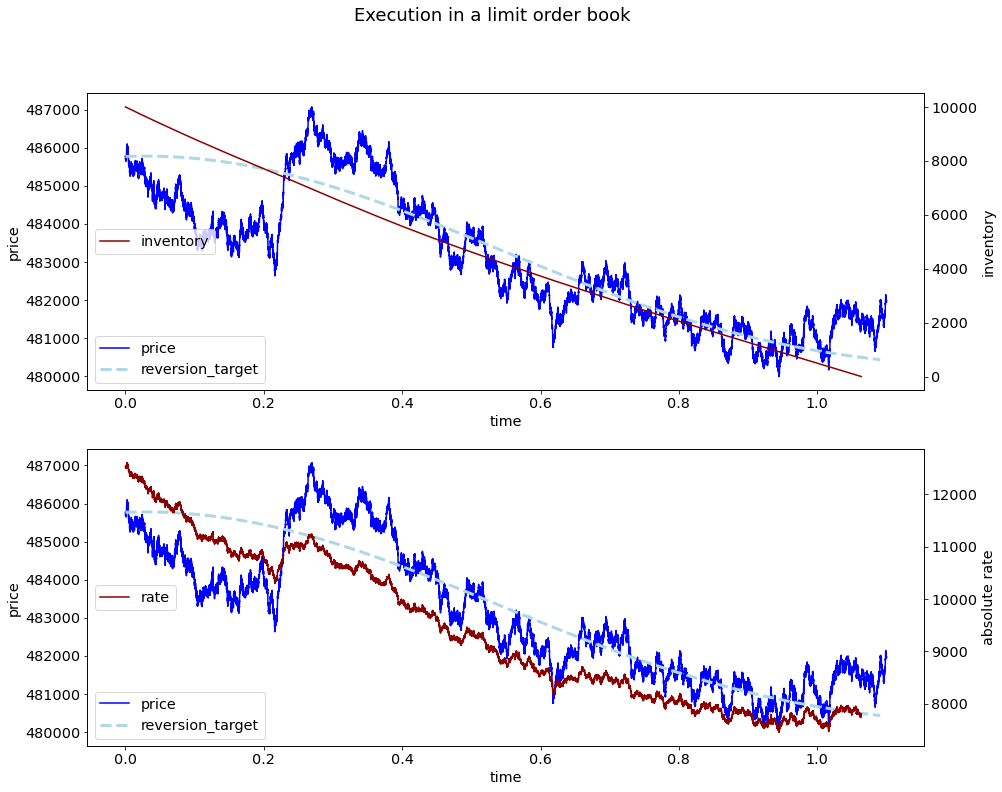} \\
				\begin{small}
			\begin{tabular}{lr | lr}
				initial inventory  $\initialInventory   $ : & 10000.0   & price process : & see equation \eqref{eq.jumpdiffusion} \\
				liquidation target $\liquidationTarget   $ : & 0.0 & coef market impact $ \coeffMarketImpact  $: & 1.35\\
				time horizon of liquidation $\timeHorizon   $ : & 1.0 & coef risk aversion $ \coeffRiskAversion   $ : & 1.15 \\
				initial price $\fundamentalPrice_0 $ : & 485777 & 	\\
			\end{tabular}
		\end{small}
	\end{tabular}
	\caption{{Upper quadrant: Historical intraday mid-price of INTC on 2019-01-22 (10am-3pm) and a sample path of the process in equation \eqref{eq.jumpdiffusion}. Lower quadrant:  Inventory trajectory and rate for the liquidation referring to the sampled price path.}}
	\label{fig.INTCliquidation}
\end{figure}
\end{center}

\begin{center}
 \begin{figure}
\centering
\begin{tabular}{c}
\includegraphics[width=.9\textwidth]{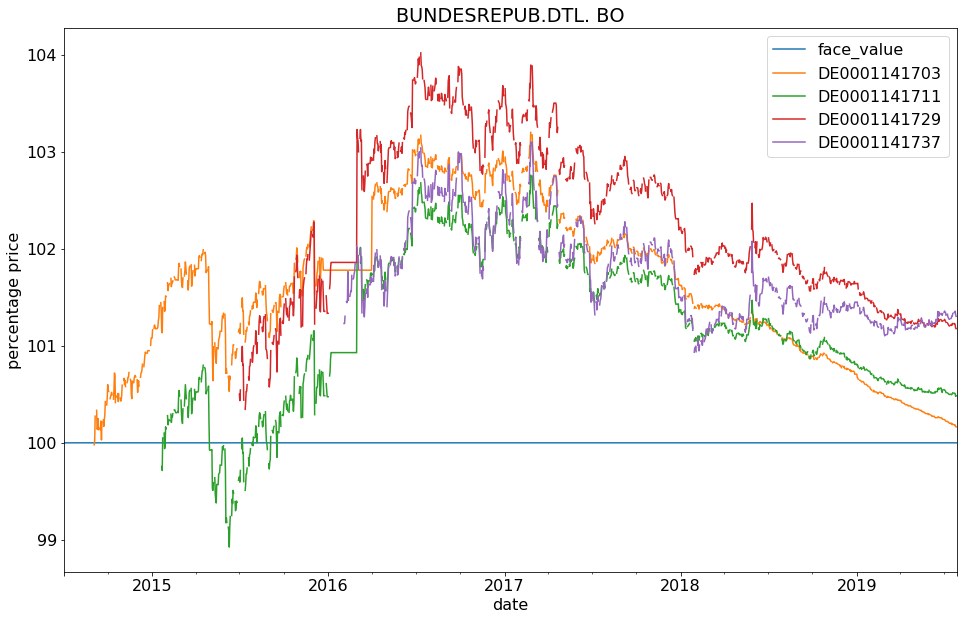}
\\
\begin{tiny}
\begin{tabular}{lllllrr}
	  \textbf{bond name} &          \textbf{ISIN} & \textbf{borrower} & \textbf{issue date} & \textbf{maturity} &  \textbf{cpn} &  \textbf{red.yield}  \\

    	BUNDESREPUB.DTL.BO 2015 ZERO 17/04/20  &  DE0001141711 &         BCKKE & 2015-01-23 &           2020-04-17 &    0.00 &           -0.6459  \\
     	BUNDESREPUB.DTL.BO 2015 1/4\% 16/10/20  &  DE0001141729 &         BCKKE & 2015-07-03 &           2020-10-16 &    0.25 &           -0.6955  \\ 
      	BUNDESREPUB.DTL.BO 2014 1/4\% 11/10/19  &  DE0001141703 &         BCKKE & 2014-09-05 &           2019-10-11 &    0.25 &           -0.4781  \\ 
       	BUNDESREPUB.DTL.BO 2016 ZERO 09/01/21 &  DE0001141737 &         BCKKE & 2016-02-05 &           2021-04-09 &    0.00 &           -0.7484  \\
\end{tabular}
\end{tiny}
\\
\\
\hline\\
\includegraphics[width=0.99\textwidth]{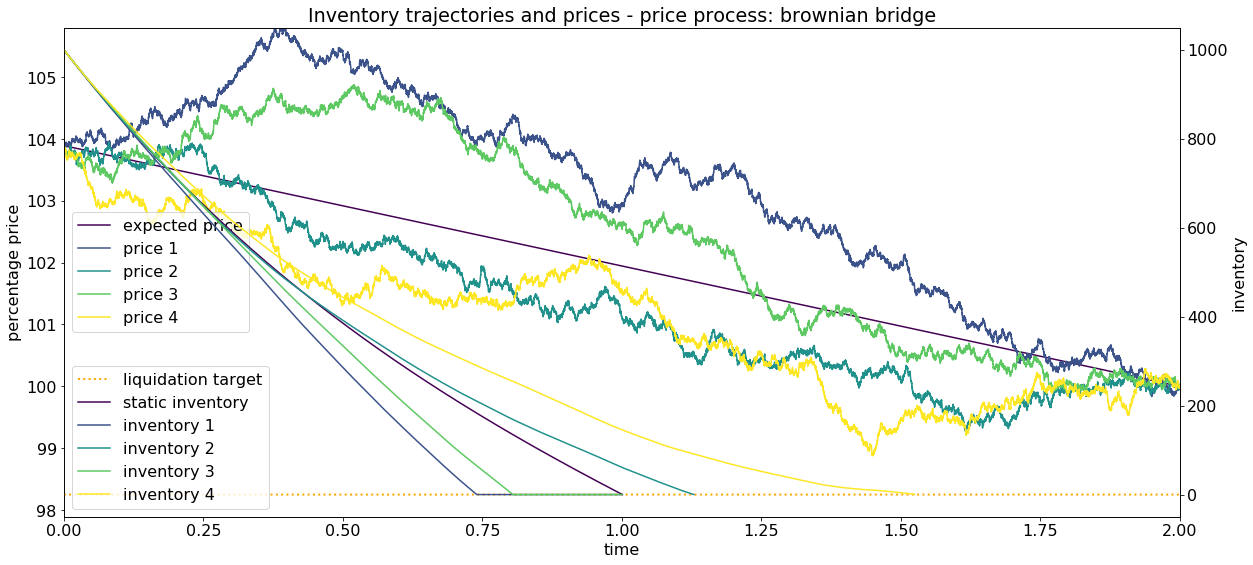} \\
		\begin{small}
			\begin{tabular}{lr | lr}
				initial inventory  $\initialInventory   $ : & 1000.0   & price process : & see equation \eqref{eq.priceProcessBond} \\
				liquidation target $\liquidationTarget   $ : & 0.0 & volatility $\sigma   $ : & 1.1642\\
				time horizon of liquidation $\timeHorizon   $ : & 1.0 & coef market impact $ \coeffMarketImpact  $: & 0.05855\\
				initial price $\fundamentalPrice_0 $ : & 103.893 & 	coef risk aversion $ \coeffRiskAversion   $ : & 0.07341 \\
			\end{tabular}
		\end{small}
	\end{tabular}
	\caption{{Upper quadrant: Historical prices of bonds with negative yields. Lower quadrant: Good trade executions in the context of a Brownian bridge.}}
	\label{fig.bonds}
\end{figure}
\end{center}

\newpage 

\addcontentsline{toc}{section}{References}
\bibliographystyle{apalike}
\bibliography{bibliography}

\begin{appendices}
\section{Euler-Lagrange equations with price paths} \label{sec.eulerLagrangeInPresenceOfPricePath}
All the statements in this appendix are pathwise, hence we leave probability out of the picture. In particular, the symbol $\fundamentalPrice$ will denote a (possibly discontinuous) single path of finite $p$-variation, for some $p\geq 1$. We consider the following functions spaces. The space $\smoothCompactlySupportedFunctions(0,\timeHorizon)$ is the space of smooth functions with compact support in the open interval $(0,\timeHorizon)$. The Sobolev space $\sobolevSpaceOneTwo(0,\timeHorizon)$ is the space of absolutely continuous functions $u$ on the closed interval $\timeWindow$ such that $u$ and its derivative  $\dot{u}$ are square integrable over $\timeWindow$. The Sobolev space  $\sobolevSpaceOneTwo(0,\timeHorizon)$ is endowed with the norm 
\begin{equation*}
\norm[u]\squared _{\sobolevSpaceOneTwo} = \int_{0}^{\timeHorizon} \left(u_t\squared + \dot{u}_t\squared \right) dt. 
\end{equation*}
The space $\sobolevSpaceOneTwoCompactSupport(0,\timeHorizon)$ is defined as the closure of $\smoothCompactlySupportedFunctions(0,\timeHorizon)$ in  $\sobolevSpaceOneTwo(0,\timeHorizon)$. 
\begin{lemma}\label{lemma.SobolevSeminorm}
 Let $f$ and $g$ be non-negative bounded measurable functions defined on $(0,\timeHorizon)$. Then,
 \begin{equation*}
  \norm[u]
  :=
  \left(
  \int_{0}^{\timeHorizon} 
  \left(
  u\squared(t)f(t)+\dot{u}\squared(t)g(t)
  \right)
  dt
  \right)^{\frac{1}{2}}
 \end{equation*}
is a seminorm on $\sobolevSpaceOneTwo(0,\timeHorizon)$.
\end{lemma}
\begin{proof}
 Notice that for $u$ in $\sobolevSpaceOneTwo(0,T)$ we can write $\norm[u]\squared = $ $\norm[u]_{f}\squared + \norm[\dot{u}]_{g}\squared$, where $\norm_{f}$ is the $\Ltwo$ norm associated with the measure $f(t)dt$, and $\norm_{g}$ is the $\Ltwo$ norm associated with the measure $g(t)dt$. Therefore, the triangle inequality is proved by the following chain of inequalities:
 \begin{equation*}
  \begin{split}
   \norm[u+v]
   \leq &
   \left( (\norm[u]_f + \norm[v]_f)\squared
   + (\norm[\dot{u}]_g + \norm[\dot{v}]_g)\squared
   \right)^{\half}
   \\
   \leq &
   \left( \norm[u]_f\squared + \norm[\dot{u}]_g\squared \right)^{\half}
   +\left( \norm[v]_f\squared + \norm[\dot{v}]_g\squared \right)^{\half},
  \end{split}
 \end{equation*}
where on the first line we used Minkowski inequality in $\Ltwo_{f}$ and $\Ltwo_{g}$, and on the second line we applied Minkowski inequality in $\ell\squared$. 
\end{proof}

\begin{defi}[``Caratheodory function'']\label{def.caratheodory}
Let $U$ be an open set in $\Rn$. The function $F:U\times \Rd \rightarrow \R$ is called Caratheodory function  if:
1. for almost every $t$ in $U$, the map $x \mapsto F(t,x)$ is continuous;
2. for every $x$ in $\Rd$, the function $t\mapsto F(t,x)$ is Lebesgue-measurable. 
\end{defi}

\begin{defi}\label{def.space-diff_caratheodory}
Let  $U$ be an open subset of $\R_+$. Let $F:U\times \Rd \rightarrow \R$ be a  Caratheodory function. We say that $F$ is space-differentiable if for all $t$ in $U$ the map $x\mapsto F(t,x)$ is in $C^{1}(\Rd)$. 
\end{defi}
In this appendix we consider Lagrangians $F$ for which the following assumption holds.
\begin{assumption}\label{assumption.decompositionOfF}
Let $F:(0,\timeHorizon)\times \R^{3} \rightarrow \R$ be a Caratheodory function. It is assumed that 
\begin{equation}\label{eq.abstractDecompositionOfLagrangian}
F(t,x_1,x_2,x_3) = x_1 x_3 + \Lagrangian (t,x_1,x_2,x_3),
\end{equation}
where $\Lagrangian$ is a space-differentiable Caratheodory function on $(0,\timeHorizon)\times \R^{3}$ such that  there exist a function $\alpha$ in $L^{1}(0,\timeHorizon)$ and a constant $\beta\geq 0$ such that 
\begin{equation}
\label{eq.growthConditionOnLarangian}
\begin{split}
\abs{\Lagrangian(t,x_1,x_2,x_3)}, 
\abs{\partial_{x_2}\Lagrangian(t,x_1,x_2,x_3)}, &
\abs{\partial_{x_3}\Lagrangian(t,x_1,x_2,x_3)}\\
\leq & \quad  \alpha(t) + \beta\Big(x_1\squared + x_2\squared + x_3\squared\Big),
\end{split}
\end{equation}
for all $t$ in $(0,\timeHorizon)$ and all $x_1$,  $x_2$, $x_3$ in $\R$.
\end{assumption}

The space variables $x=(x_1,x_2,x_3)$ are alternatively labelled as $(x_1,x_2,x_3) = (\fundamentalPrice,q,r)$. Hence, we interpret $x_1$ as the placeholder for the variable $\fundamentalPrice$ denoting the fundamental price, we interpret $x_2$ as the placeholder for the variable $\inventory$ denoting the inventory, and we interpret $x_3$ as the placeholder for the variable $r$ denoting the rate of execution. In the following, we adopt these letters for the space variables; hence, in particular $\partial_q \Lagrangian$ denotes the derivative of $\Lagrangian$ with respect to the variable $x_2$, and $\partial_r \Lagrangian$ denotes the derivative of $\Lagrangian$ with respect to the variable $x_3$.

\begin{remark}
When $\Lagrangian$ does not depend on $x_1$, the  decomposition in equation \eqref{eq.abstractDecompositionOfLagrangian} has the form of  the decomposition in equation \eqref{eq.decompositionOfLagrangian}. However, we do not restrict to that case in this appendix. In particular, our treatment encompasses the case of the functional used in \cite{GS11opt}, which we discuss in Section \ref{sec.varInspiredRiskCriterion}. 
\end{remark}

Let $\fundamentalPrice = (\fundamentalPrice_t)_{0\leq t \leq \timeHorizon}$ be a path on $\timeWindow$ and assume that $\fundamentalPrice$ is of finite $p$-variation for some $p\geq 1$. Let $F$ be a Caratheodory function satisfying Assumption \ref{assumption.decompositionOfF}.  Let $\costFunctional$ be the functional 
\begin{equation}\label{eq.definitionOfCostFunctionalForBoundaryValueProblem}
\begin{split}
\costFunctional (\eta) := \intZeroTimeHorizon F(t,\fundamentalPrice_t,\eta_t,\dotEta_t) dt , 
\qquad \eta \in \sobolevSpaceOneTwo(0,\timeHorizon).
\end{split}
\end{equation}
Let $q_0$ be in $\sobolevSpaceOneTwo(0,\timeHorizon)$ and consider the affine space
\begin{equation*}
 q_0 +\sobolevSpaceOneTwoCompactSupport(0,\timeHorizon) = \left\lbrace
 q_0 + u: \, u \in \sobolevSpaceOneTwoCompactSupport(0,\timeHorizon)
 \right\rbrace.
\end{equation*}
We study the minimisation problem
\begin{equation}
\label{eq.minimisationBoundaryValueProblem}
\inf\left\lbrace
\costFunctional (\eta): \, \eta \in q_0 +\sobolevSpaceOneTwoCompactSupport(0,\timeHorizon)
\right\rbrace.
\end{equation}

\begin{prop}[``Weak form of Euler-Lagrange equation'']\label{prop.weakFormEulerLagrange}
Let $F$ be as in Assumption \ref{assumption.decompositionOfF}. Let $\Lagrangian(t,x_1,x_2,x_3):= F(t,x_1,x_2,x_3)-x_1 x_3$.  Assume that $q$ is a minimiser for \eqref{eq.minimisationBoundaryValueProblem}. Then, for all $\psi$ in $\smoothCompactlySupportedFunctions(0,\timeHorizon)$ it holds 
\begin{equation}
\label{eq.weakFormEulerLagrange}
\intZeroTimeHorizon \big( \fundamentalPrice_t + \partial_r \Lagrangian (t,S_t,q_t,\dot{q}_t) \big ) \dot{\psi}_t dt 
= 
- \intZeroTimeHorizon \psi_t \partial_q \Lagrangian (t,S_t,q_t,\dot{q}_t) \, dt .
\end{equation}

Conversely, if $(x_2, x_3)\mapsto F(t,x_1,x_2,x_3)$ is convex for all $t$ and all $x_1$, then \eqref{eq.weakFormEulerLagrange} is a sufficient condition on $q \in q_0 +  \sobolevSpaceOneTwoCompactSupport(0,\timeHorizon)$ for being a minimiser for \eqref{eq.minimisationBoundaryValueProblem}.
\end{prop}
\begin{proof}
For $\psi$ in $\smoothCompactlySupportedFunctions(0,\timeHorizon)$ it holds 
\begin{equation*}
\begin{split}
\lim_{\epsilon \downarrow 0} \frac{1}{\epsilon }\intZeroTimeHorizon \big[
	\Lagrangian &(t, S_t, q_t + \epsilon \psi_t, \dot{q}_t + \epsilon \dot{\psi}_t ) -  \Lagrangian (t,S_t,q_t, \dot{q}_t  )
\big] dt \\
& = \intZeroTimeHorizon \big[
	\psi_t \partial_q  \Lagrangian (t,S_t,q_t, \dot{q}_t  ) + \dot{\psi}_t \partial_r  \Lagrangian (t,S_t,q_t, \dot{q}_t  )
\big] dt, 
\end{split}
\end{equation*}
owing to assumption \eqref{eq.growthConditionOnLarangian} on the growth of $\Lagrangian$ and its space derivatives. For all $\psi$ in $\smoothCompactlySupportedFunctions(0,\timeHorizon)$ and all $\epsilon>0$ we have that 
\begin{equation*}
\costFunctional \big( q+\epsilon\psi \big) \geq \costFunctional \big( q\big),
\end{equation*}
and thus necessity of \eqref{eq.weakFormEulerLagrange} is established.

Conversely, assume that $(x_2, x_3)\mapsto F(t,x_1,x_2,x_3)$ is convex for all $t$ and all $x_1$. Assume that $q$ in $q_0 +  \sobolevSpaceOneTwoCompactSupport(0,\timeHorizon)$ satisfies equation \eqref{eq.weakFormEulerLagrange}. Let $\eta$ be arbitrary in $q_0 +  \sobolevSpaceOneTwoCompactSupport(0,\timeHorizon)$, and let $t$ be arbitrary in $\timeWindow$. By convexity it holds
\begin{equation*}
\begin{split}
 F(t,S_t,\eta_t,\dot{\eta}_t)
 \geq
 &
 F(t,S_t,q_t,\dot{q}_t)
 + \partial_q \Lagrangian(t,S_t,q_t,\dot{q}_t)(\eta_t - q_t)
 \\
 &
 + 
 \left[
 S_t + \partial_r \Lagrangian(t,S_t,q_t,\dot{q}_t)
 \right]
 (\dot{\eta}_t - \dot{q}_t) 
 \end{split}.
\end{equation*}
Notice that $\eta - q$ is in $\sobolevSpaceOneTwoCompactSupport(0,\timeHorizon)$ and thus it is approximated by  $\psi$ in $\smoothCompactlySupportedFunctions(0,\timeHorizon)$. Therefore, if we integrate the latter inequality in $dt$ and use \eqref{eq.weakFormEulerLagrange}, we conclude 
$ \int_{0}^{T}F(t,S_t,\eta_t,\dot{\eta}_t) dt
 \geq
 \int_{0}^{T}
 F(t,S_t,q_t,\dot{q}_t)dt$. 
\end{proof}

As an application of the sufficient condition for optimality established in Proposition \ref{prop.weakFormEulerLagrange}, we derive the optimal static inventory trajectory in the case of linear temporary market impact and quadratic inventory cost. 

\begin{corol}\label{corol.staticOptimalSolIC}
Let $\fundamentalPrice$ be a price process as defined in Definition \ref{def.priceprocess}. Let $\initialInventory$, $\coeffMarketImpact$ and $\coeffRiskAversion$ be positive real numbers, and let $\ratioAversionOverImpact:=\coeffRiskAversion/\coeffMarketImpact$. Let $\alpha$ and $K$ be as in Proposition \ref{prop.goodExecutionIC}.  For $0\leq t \leq \timeHorizon$ define
\begin{equation*}
\begin{split}
\inventory_t:=  
\big(1-\alpha(t)\big) \initialInventory +\alpha(t)\liquidationTarget
- \frac{1}{2\coeffMarketImpact\squared}
\int_{0}^{t} \cosh\big(\ratioAversionOverImpact (t-u)\big)\Expectation[\fundamentalPrice_u]du 
+ K \sinh(\ratioAversionOverImpact t).
\end{split}
\end{equation*}
Then, the function $t\mapsto \inventory_t$ is a minimiser for
\begin{equation*}
 \inf
 \left\lbrace
 \Expectation
 \left[
 \int_{0}^{\timeHorizon}
 \left(
 \fundamentalPrice_t \dot{\eta}_t + \coeffMarketImpact\squared \dot{\eta}_{t}\squared + 
 \coeffRiskAversion\squared \eta_{t}\squared
 \right)
 dt
 \right] : \, \, 
 \eta \in \spaceInventoryTrajectoriesStatic^{0,\initialInventory}
 \right\rbrace.
\end{equation*}
\end{corol}
\begin{proof}
 Let 
 \begin{equation}\label{eq.FinStaticIC}
  F(x_1,x_2,x_3):= x_1 x_3 + \coeffMarketImpact\squared x_{3}\squared + \coeffRiskAversion\squared x_{2}\squared,
 \end{equation}
and notice that $(x_2, x_3)\mapsto F(x_1,x_2,x_3)$ is convex for all $x_1$.
For all $\eta$ in $\spaceInventoryTrajectoriesStatic^{0,\initialInventory}$ it holds
\begin{equation*}
 \Expectation
 \left[
 \int_{0}^{\timeHorizon}
 \left(
 \fundamentalPrice_t \dot{\eta}_t + \coeffMarketImpact\squared \dot{\eta}_{t}\squared + 
 \coeffRiskAversion\squared \eta_{t}\squared
 \right)
 dt
 \right]
 =
 \int_{0}^{\timeHorizon}
 F\big(\Expectation[\fundamentalPrice_t], \eta_t,\dot{\eta}_t\big) dt.
\end{equation*}
Let $\inventory$ be as in the statement. We have that
\begin{equation*}
 \frac{d}{dt}\Big(
 2\coeffMarketImpact\squared \dot{\inventory}_t +\Expectation[\fundamentalPrice_t]
 \Big)
 =
 2\coeffRiskAversion\squared \inventory_t.
\end{equation*}
Hence, an application of the integration-by-parts formula shows that  $\inventory$ satisfies \eqref{eq.weakFormEulerLagrange} with $F$ as in \eqref{eq.FinStaticIC} and $\fundamentalPrice_t$ replaced by $\Expectation[\fundamentalPrice_t]$.
\end{proof}

Equation \eqref{eq.weakFormEulerLagrange} is the weak form of the Euler-Lagrange equation. In the next  paragraph, we introduce the pathwise integration with respect to the path $\fundamentalPrice$. This will allow to move from the condition in equation \eqref{eq.weakFormEulerLagrange} to the stronger formulation of the Euler-Lagrange equation. 

\subsubsection*{Young integration and integration-by-parts formula}
 Let $\timeWindow$ be the closed time interval from time zero to the time horizon $\timeHorizon>0$. We consider partitions $\partition$ of $\timeWindow$. 
 A partition $\partition$ is simultaneously considered  as the finite collection of points and as the finite collection of adjacent subintervals. Given a partition $\partition$ of $\timeWindow$ and a time instant $t$ in $\timeWindow$,  we adopt the following notational convention: 
 \begin{equation}\label{eq.notationAboutPartition}
 t\derivative:=  \inf \lbrace u \in \partition : \, u> t \rbrace .
 \end{equation} 
 The mesh-size $\abs{\partition}$ of the partition $\partition$ is defined as 
 \begin{equation*}
 \abs{\pi} :=  \sup \lbrace \abs{u\derivative - u}: \, u \in \pi \rbrace.
 \end{equation*}
 
Let $f$ be a function on $\timeWindow$. If $s$ and $t$ are  in $\timeWindow$,  we let $f_{s,t}:=f_t - f_s$ denote  the increment of $f$ from $s$ to $t$. The resulting two-parameter function $f=f_{s,t}$ is additive, in that for all $s\leq u\leq t$ it holds $f_{s,u} + f_{u,t} = f_{s,t}$. For more general two-parameter functions  defined on $\simplex$ we can relax the additivity and consider the following property. 
\begin{defi}
Let $f=f(s,t)$ be a function on $\simplex$. We say that $f$ is super-additive if  for all $s\leq u\leq t$ it holds 
\begin{equation*}
f(s,u) + f(u,t) \leq f(s,t). 
\end{equation*}
\end{defi}

\begin{lemma}\label{lemma.noCorrectionTermsInIntegrationByParts}
Let $f$ be a super-additive   function on $\simplex$ such that it is null and uniformly continuous on the diagonal, i.e. $f(t,t) \equiv 0$ and 
	\begin{equation}\label{eq.uniformContinuityOnTheDiagonal}
	\lim_{\epsilon\downarrow 0} \sup \left\lbrace f(s,t): \, \abs{t-s} \leq \epsilon\right\rbrace
	=0.
	\end{equation}
Let $g$ be a  super-additive function on $\simplex$. Let $0< a \leq 1$ and $0<b<1$ be exponents such that $a+b>1$. Then,
 \begin{equation*}
 \lim_{\abs{\partition}\downarrow 0 } \sum_{u\in\partition} f^{a} g^{b} (u,u\derivative) = 0, 
 \end{equation*}
 where the limit is taken along arbitrary sequences of partitions of $\timeWindow$  with mesh-size tending to zero. 
\end{lemma}
\begin{proof}
Let $0<\epsilon<a$ be such that $a-\epsilon + b =1$. Then,
\begin{equation*}
\begin{split}
\sum_{u\in\partition} f^a g^b (u,u\derivative) \leq  & 
\left(\sup_{u\in\partition} f ^{\epsilon} (u,u\derivative)\right) 
\cdot 
\sum_{u\in\partition} f^{a-\epsilon}g^{b} (u,u\derivative) \\
\leq & 
\left(\sup_{u\in\partition} f ^{\epsilon} (u,u\derivative)\right) 
\cdot 
\left( 
	\sum_{u\in\partition} f(u,u\derivative) 
\right)^{a-\epsilon}
\left( 
\sum_{u\in\partition} g(u,u\derivative) 
\right)^{b} \\
\leq & 
\left(\sup_{u\in\partition} f ^{\epsilon} (u,u\derivative)\right) 
\cdot f^{a-\epsilon}(0,\timeHorizon) g^{b}(0,\timeHorizon). 
\end{split}
\end{equation*}
On the second line we have applied H\"older inequality and on the third line we have used super-additivity. We conclude by recalling the uniform continuity of equation \eqref{eq.uniformContinuityOnTheDiagonal}. 
\end{proof}

\begin{lemma}
\label{lemma.integrationByPartsYoungIntegral}
Let $\eta$ be an absolutely continuous path on the closed interval $\timeWindow$. Let $\fundamentalPrice$ be a path on $\timeWindow$ of finite $p$-variation for some $p\geq 1$. Then, for all $0\leq s \leq t \leq \timeHorizon$  the limit 
\begin{equation*}\label{eq.definitionOfYoungIntegral}
 \lim_{\abs{\partition}\downarrow 0 } \sum_{u\in\partition} \eta_{u} \fundamentalPrice_{u,u\derivative}
\end{equation*}
exists  and is the same along any sequence of partitions of $[s,t]$. Such a limit defines the Young integral 
\begin{equation*}
\int_{s}^{t} \eta_u d\fundamentalPrice_u .
\end{equation*}
 Moreover, the following integration-by-parts formula holds 
 \begin{equation*}
 \int_{s}^{t} \eta_u d\fundamentalPrice_u  + \int_{s}^{t} \fundamentalPrice_u d\eta_u = \eta_t \fundamentalPrice_t - \eta_s\fundamentalPrice_s,  
 \end{equation*}
 where the second integral on the left hand side is the Stieltjes integral of $\fundamentalPrice$ against $\eta$. 
\end{lemma}
\begin{proof}
Let $\partition$ be an arbitrary partition of $[s,t]$. For all $u$ in $\partition$ we have 
\begin{equation}\label{eq.discreteSummationByParts}
\eta_u \fundamentalPrice_{u,u\derivative} + \fundamentalPrice_{u}\eta_{u,u\derivative}
= \eta_{u\derivative}\fundamentalPrice_{u\derivative} - \eta_{u}\fundamentalPrice_{u} - \eta_{u,u\derivative} \fundamentalPrice_{u,u\derivative}. 
\end{equation}
Consider first the right hand side of equation \eqref{eq.discreteSummationByParts}. Consider the first two summands on the right hand side of   equation \eqref{eq.discreteSummationByParts}. If we sum over $u$ in $\partition$ we have the telescopic sum 
\begin{equation*}
\sum_{u\in\partition} \left(\eta_{u\derivative}\fundamentalPrice_{u\derivative} - \eta_u \fundamentalPrice_u\right) = \eta_t\fundamentalPrice_t - \eta_s \fundamentalPrice_s,
\end{equation*}
and this does not depend on the partition $\partition$ of $[s,t]$. 
The third summand on the right hand side of   equation \eqref{eq.discreteSummationByParts} can be bound as follows:
\begin{equation*}
\abs{\eta_{u,u\derivative}\fundamentalPrice_{u,u\derivative}}
 \leq \pvarNormInterval[\eta]{1}{[u,u\derivative]} \pvarNormInterval[\fundamentalPrice]{p}{[u,u\derivative]}.
\end{equation*}
Hence, by Lemma \ref{lemma.noCorrectionTermsInIntegrationByParts}, we have
\begin{equation*}
\lim_{\abs{\partition}\downarrow 0 }\sum_{u\in\partition} \eta_{u,u\derivative}\fundamentalPrice_{u,u\derivative} = 0. 
\end{equation*}
Here, we have applied Lemma \ref{lemma.noCorrectionTermsInIntegrationByParts} with $f(s,t) = \pvarNormInterval[\eta]{1}{[s,t]}$, $g(s,t) = \pvarNormInterval[\fundamentalPrice]{p}{[s,t]}^{p}$, $a=1$, and $b=1/p$. 
Therefore, if we apply the operator $\lim_{\abs{\partition}\downarrow 0 }\sum_{u\in\partition}$ to the right hand side of equation \eqref{eq.discreteSummationByParts}, we obtain 
\begin{equation*} \label{eq.meshsizeShrinkingOnRHSofIntegrationByParts}
\lim_{\abs{\partition}\downarrow 0 }\sum_{u\in\partition} \left(
 \eta_{u\derivative}\fundamentalPrice_{u\derivative} - \eta_{u}\fundamentalPrice_{u} - \eta_{u,u\derivative} \fundamentalPrice_{u,u\derivative}
 \right) = \eta_t\fundamentalPrice_t - \eta_s \fundamentalPrice_s. 
\end{equation*}
Consider now the left hand side of  equation \eqref{eq.discreteSummationByParts}. The limit 
\begin{equation*}
\lim_{\abs{\partition}\downarrow 0 } \sum_{u\in\partition} \fundamentalPrice_{u}\eta_{u,u\derivative}
\end{equation*}
is the  Stieltjes integral $\int_{s}^{t} \fundamentalPrice_u d\eta_u  $ and thus converges and is the same along any sequence of partitions with vanishing meshsize. Therefore, if we apply the operator $\lim_{\abs{\partition}\downarrow 0 }\sum_{u\in\partition}$ to both side of equation \eqref{eq.discreteSummationByParts} we must have that the limit in equation \eqref{eq.definitionOfYoungIntegral} exists and is equal to 
\begin{equation*}
\eta_t\fundamentalPrice_t - \eta_s \fundamentalPrice_s - \int_{s}^{t} \fundamentalPrice_u d\eta_u .
\end{equation*}
This concludes the proof of the lemma. 
\end{proof}

\subsubsection*{Strong form of Euler-Lagrange equation}
Relying on the Young integral introduced in the previous  paragraph, we now formulate the strong form of the Euler-Lagrange equation associated with the minimisation problem in equation \eqref{eq.minimisationBoundaryValueProblem}. 

\begin{defi}\label{defi.solutionToSecondOrderRDE}
Let $c$ be a constant. Let $\fundamentalPrice$ be  a path on $\timeWindow$ and assume that $\fundamentalPrice$ is of finite $p$-variation for some $p\geq 1$. Let $b=b(t,x)$ be a Caratheodory function on $(0,\timeHorizon)\times \R^{3}$ such that 
\begin{equation*}
\abs{b(t,x)} \leq \alpha(t) + \beta \abs{x}\squared,
\end{equation*}
for some integrable function $\alpha$ in $L^1(0,\timeHorizon)$ and some constant $\beta \geq 0$. We say that the function $q$ in $\sobolevSpaceOneTwo(0,\timeHorizon)$ solves the equation
\begin{equation}\label{eq.secondOrderRDE}
\begin{cases}
dq_t =& r_t dt \\
dr_t = & b(t,S_t,q_t,\dot{q}_t) dt - c d\fundamentalPrice,
\end{cases}
\end{equation}
if for all $\eta$ in $\smoothCompactlySupportedFunctions(0,\timeHorizon)$  and all $0\leq s\leq t \leq \timeHorizon$ it holds
\begin{equation}\label{eq.meaningOfSolutionToSecondOrderRDE}
\int_{s}^{t} \dotEta_u dq_u = \eta_t \dot{q}_t - \eta_s \dot{q}_s - \int_{s}^{t} \eta_u b(u,q_u,\dot{q}_u) du + c\int_{s}^{t} \eta_u d\fundamentalPrice_u,
\end{equation}
where the last integral on the right hand side is the Young integral introduced in Lemma \ref{lemma.integrationByPartsYoungIntegral}. 
\end{defi}

\begin{lemma}\label{lemma.characterisationOfSolutionToSecondOrderRDE}
Assume the setting of Definition \ref{defi.solutionToSecondOrderRDE}. Then, $q$ solves equation \eqref{eq.secondOrderRDE} if and only if the function $f_t:= \dot{q}_t + c\fundamentalPrice_t$ is absolutely continuous with derivative 
\begin{equation*}
\dot{f}_t = b(t,S_t,q_t,\dot{q}_t).
\end{equation*}
\end{lemma}
\begin{proof}
Apply the integration-by-parts established in Lemma \ref{lemma.integrationByPartsYoungIntegral}  to formula \eqref{eq.meaningOfSolutionToSecondOrderRDE}, and observe  that -- for   all $\eta$ in $\smoothCompactlySupportedFunctions(0,\timeHorizon)$ --  formula \eqref{eq.meaningOfSolutionToSecondOrderRDE} is equivalent to 
\begin{equation*}
\intZeroTimeHorizon \Big(\dot{q}_t + c \fundamentalPrice_t \Big)\dotEta_t dt
 = - \intZeroTimeHorizon \eta_t b(t,S_t,q_t,\dot{q}_t) dt. 
\end{equation*}
\end{proof}

\begin{prop}[``Strong form of Euler-Lagrange equation'']\label{prop.eulerLagrangeNecessity}
 Let $\fundamentalPrice$ be  a path on $\timeWindow$ and assume that $\fundamentalPrice$ is of finite $p$-variation for some $p\geq 1$. Consider a Lagrangian $F$ satisfying Assumption \ref{assumption.decompositionOfF}. Assume that the space-differentiable Caratheodory function $\Lagrangian(t,x):=F(t,x)-x_1 x_3$ is in $C\squared(\timeWindow \times \R^{3})$ and such that 
 \begin{equation*}
\partial_{x_3} \Lagrangian (t,x_1,x_2,x_3) = \frac{1}{c}x_3 + \ell(t,x_2),
 \end{equation*}
 for some non-zero constant $c$ in $\R$ and  some function $\ell$ in $C^{1}(\timeWindow \times \R)$.\footnote{This assumption on the form of the partial derivative  $\partial_r \Lagrangian$ is satisfied in all three examples that we consider in the paper, namely by the Lagrangians $F(t,S,q,r) = rS + \Lagrangian(t,S,q,r)$ in equations \eqref{eq.LagrangianIC}, \eqref{eq.LagrangianTimeIC} and \eqref{eq.LagrangianVaR}.} Assume that $q$ is a minimiser for \eqref{eq.minimisationBoundaryValueProblem}. Then, $q$ solves the Euler-Lagrange equation 
 \begin{equation}\label{eq.strongFormEulerLagrangeEquation}
 \begin{cases}
 dq_t =&r_t dt \\
 dr_t = & \Big( \partial_q \Lagrangian (t,S_t,q_t,\dot{q}_t) - \partial\squared _{t,r} \Lagrangian(t,S_t,q_t,\dot{q}_t) - \dot{q}_t \partial\squared _{q,r} \Lagrangian(t,S_t,q_t,\dot{q}_t) \Big)c dt - cd\fundamentalPrice_t,
 \end{cases}
 \end{equation}
 in the sense of Definition \ref{defi.solutionToSecondOrderRDE}.
\end{prop}
\begin{proof}
Under the stated assumptions on the form of the partial derivative $\partial_r \Lagrangian$, the condition in equation \eqref{eq.weakFormEulerLagrange} reads
\begin{equation*}
\intZeroTimeHorizon \Big[ \fundamentalPrice_t + \frac{1}{c}\dot{q}_t + \ell(t,q_t)\Big]\dot{\psi}_t dt 
= - \intZeroTimeHorizon \psi_t \partial_q \Lagrangian (t,S_t,q_t,\dot{q}_t) dt . 
\end{equation*}
This is equivalent to 
\begin{equation*}
\intZeroTimeHorizon \Big[ c\fundamentalPrice_t + \dot{q}_t \Big]\dot{\psi}_t dt 
= - c\intZeroTimeHorizon \Big[\partial_q \Lagrangian (t,S_t,q_t,\dot{q}_t) - \partial_t \ell (t,q_t) - \dot{q}_t \partial_q \ell  (t,q_t) \Big]\psi_t  dt . 
\end{equation*}
Notice that $\partial_t \ell = \partial\squared _{t,r} \Lagrangian $ and $\partial_q \ell = \partial\squared _{q,r} \Lagrangian$. Hence, Lemma \ref{lemma.characterisationOfSolutionToSecondOrderRDE} concludes the proof. 
\end{proof}
\end{appendices}
\end{document}